\documentclass[journal]{IEEEtran}

\usepackage[utf8]{inputenc}
\usepackage{amsmath,amssymb,amsthm}
\usepackage{bm}
\usepackage{mathtools}
\usepackage{graphicx}
\usepackage{bbm}
\usepackage{algorithm}
\usepackage{algorithmic}
\usepackage{booktabs}
\usepackage{enumitem}
\usepackage{color}
\usepackage{float}
\usepackage[colorlinks]{hyperref}
\usepackage[caption=false,font=footnotesize]{subfig}
\usepackage{graphicx}

\newtheorem{theorem}{Theorem}

\newtheorem{remark}{Remark}

 \title{
Information-Theoretic Limits of\\Fluid Antenna-Aided ISAC with Representation Information Constraints}

\author{Farshad~Rostami~Ghadi,~\IEEEmembership{Senior~Member},~\textit{IEEE}, 
            Kai-Kit~Wong,~\IEEEmembership{Fellow},~\textit{IEEE}, 
	    F.~Javier~L\'opez-Mart\'inez,~\IEEEmembership{Senior~Member},~\textit{IEEE}, 
	    Zhentian~Zhang,~\IEEEmembership{Member},~\textit{IEEE},\\ 
	    Hyundong~Shin,~\IEEEmembership{Fellow},~\textit{IEEE}, 
	    and~Christos~Masouros, \IEEEmembership{Fellow}, \textit{IEEE}
\vspace{-7mm}
}

\makeatletter
\def\blfootnote{\xdef\@thefnmark{}\@footnotetext}
\makeatother
\begin{document}
\maketitle
\blfootnote{The work of F. J. L\'opez-Mart\'inez is supported by grant PID2023-149975OB-I00 (COSTUME) funded by MICIU/AEI/10.13039/501100011033 and by ERDF/EU, and by the European Commission as part of the 6G-LEADER Project under grant no. 101192080. The work of F. Rostami Ghadi is supported by the European Union's Horizon 2022 Research and Innovation Programme under Marie Sk\l odowska-Curie Grant No. 101107993. The work of  K. K. Wong is supported by the Engineering and Physical Sciences Research Council (EPSRC) under Grant EP/W026813/1. The work of H. Shin is supported by the National Research Foundation of Korea (NRF) grant funded by the Korean government (RS-2025-00556064 and RS-2025-25442355), and by the Ministry of Science and ICT (MSIT), Korea, under the ITRC (Information Technology Research Center) support program (IITP-2025-RS-2021-II212046), supervised by the IITP (Institute for Information \& Communications Technology Planning \& Evaluation).}

\blfootnote{\noindent F. Rostami Ghadi and F. J. L\'opez-Mart\'inez are with the Department of Signal Theory, Networking and Communications, Research Centre for Information and Communication Technologies (CITIC-UGR), University of Granada, 18071, Granada, Spain (e-mail: $\rm \{f.rostami, fjlm\}@ugr.es$).}
\blfootnote{\noindent K. K. Wong is with the Department of Electronic and Electrical Engineering, University College London, WC1E 7JE, London, United Kingdom, and also with the Department of Electronic Engineering, Kyung Hee University, Yongin-si, Gyeonggi-do 17104, Republic of Korea (e-mail: $\rm kai\text{-}kit.wong@ucl.ac.uk$).}
\blfootnote{\noindent Zhentian Zhang is with the National Mobile Communications Research Laboratory, Southeast University, Nanjing, 210096, China and he is also with Hong Kong Metropolitan University, Hong Kong SAR, China (e-mail: $\rm zhentianzhangzzt@gmail.com$).}
\blfootnote{\noindent H. Shin is with the Department of Electronics and Information Convergence Engineering, Kyung Hee University, Yongin-si, Gyeonggi-do 17104, Republic of Korea (e-mail: $\rm hshin@khu.ac.kr$).}
\blfootnote{\noindent  C. Masouros is with the Department of Electronic and Electrical Engineering, University College London, WC1E 7JE, London, United Kingdom (e-mail: $\rm c.masouros@ucl.ac.uk$).}
	
\blfootnote{\noindent Corresponding Author: Farshad Rostami Ghadi.}

\begin{abstract}
This paper characterizes the fundamental limits of integrated sensing and communication (ISAC) when the transmitter is subject to a representation-constrained encoder modeled via a mutual information (MI) bottleneck and the receiver employs a fluid antenna system (FAS). Specifically, the message is first encoded into an ideal Gaussian waveform and then mapped by a stochastic encoder into a finite capacity representation that serves as the channel input in the adopted analytical abstraction, while the FAS receiver selects the port with the most favorable channel conditions. From an information-theoretic perspective, we show that the representation bottleneck is equivalent to an additive representation noise, which reduces both the communication and sensing signal-to-noise ratios (SNRs) at the selected port. We then derive the resulting ISAC capacity-distortion region and establish tight converse and achievability bounds under general fading models, including Jakes correlated channels. Leveraging the spatial degrees of freedom (DoF) characterization of the Jakes' model, we prove that the port selection gain is fundamentally constrained by the physical length of the FAS region: the effective diversity order equals the numerical rank of the Jakes' correlation matrix and increases only with the FAS length. {\color{black}Consequently, enlarging the FAS length allows the selected communication SNR to approach the representation-limited ceiling, driving the achievable rate toward its fundamental bound. The sensing mean square error (MSE) similarly approaches its bound when the selected sensing gain also increases with the FAS aperture.} Numerical results corroborate the analysis and scaling laws.
\end{abstract}


\vspace{-2mm}
\section{Introduction}
\subsection{Context}
\IEEEPARstart{I}{ntegrated} sensing and communication (ISAC) is envisioned as a core capability for sixth-generation (6G) mobile networks, facilitating simultaneous data transmission and environmental perception within a unified physical layer \cite{zhang2021over,liu2023seventy}. Existing ISAC analyses mostly assume ideal transceiver processing in which signals can be represented and manipulated with unlimited precision. In practice, modern transceivers increasingly rely on artificial intelligence (AI)-based modules such as neural encoders and feature extractors. 

These modules possess finite representational capacity and cannot preserve all information contained in the ideal baseband signal. From an information-theoretic perspective, this limitation can be naturally modeled as a representation-constrained bottleneck in the transmit chain, whereby each symbol must be compressed into a latent variable with restricted information content \cite{ghadi2025info}. Such a representation constraint, which can be interpreted as arising from learning processing modules, fundamentally alters the effective signal quality seen by both communication and sensing receivers and introduces new performance ceilings absent in classical idealistic ISAC models.

A defining characteristic of ISAC systems is the intrinsic trade-off between communication and sensing performance, as both functions share common physical resources such as power, bandwidth, and waveform degrees of freedom (DoF). Enhancing communication reliability normally would reduce the fidelity of sensing observations, while prioritizing sensing accuracy can limit the achievable rate. These conflicting objectives have motivated extensive information-theoretic studies of ISAC capacity-distortion trade-offs, which characterize the fundamental limits of jointly transmitting information while estimating environmental parameters. However, most existing results rely on ideal waveform generation and signal representation assumptions, and as a consequence, do not capture the impact of practical information-theoretic transceiver constraints. When waveform generation is subjected to a finite-capacity representation-constrained bottleneck, the classical ISAC trade-off is fundamentally reshaped; even under favorable channel conditions, the achievable communication rate and sensing accuracy are bounded by the information that can be preserved through the representation-constrained encoding process, which may arise from learning modules.


{\color{black}On the other hand, fluid antenna system (FAS) has recently emerged as a hardware-agnostic system concept that considers the antenna as a physical-layer resource incorporating electromagnetic reconfiguration of the antenna into the system design \cite{Lu-2025,new2025tot,hong2025contemporary,new2025flar,wu2024flu}. In a basic port-based architecture, $L (\gg 1)$ candidate ports are distributed over a given spatial region (aperture), but only one port is selected and connected to a single radio frequency (RF) chain per channel use \cite{wong2021fas}.\footnote{In general FAS settings, selecting more than one port is possible and the selected ports can also be connected to multiple RF chains.} As a consequence, the FAS receiver can sample the channel at different spatial locations very finely and select the port with the most favorable channel realization. The FAS architecture differs from a conventional multiple fixed-antenna system, in which the physical locations of the radiating elements are far apart and remain unchanged. Conventional antenna arrays thus obtain spatial gains from multiple fixed elements through beamforming, combining, or antenna selection, whereas a port-based FAS exploits spatial DoF by changing the active location within its prescribed region. In the single RF chain configuration considered here, this flexibility is obtained without simultaneously activating all $L$ candidate ports.


In practice, FAS may be realized using movable elements \cite{zhu2026tut,zhu2024movable,shao2026,zhu2024historical}, liquid conductors \cite{shen2024design,Shamim-2025}, tunable and metasurface structures   \cite{Zhang-jsac2026,Liu-2025arxiv,liu2025meta}, or reconfigurable pixels \cite{zhang2024pixel,tong-2025pixel,Wong-wc2026}.  In \cite{tong2025designs}, Tong {\em et al.}~discussed the pros and cons of different implementation technologies for realizing the FAS concept. The choice of a specific realization affects practical properties such as radiation efficiency, insertion loss, switching or actuation speed, control power, mechanical reliability, and electromagnetic coupling between candidate states, which are classically abstracted for the sake of generality.


    
}
	
The capability of FAS provides several unique advantages that are attractive for ISAC. First of all, port selection exploits location-dependent channel variations, yielding diversity gains akin to multi-fixed-antenna systems but without the associated RF complexity. Second, selecting the strongest among many spatially distinct channel realizations creates an inherent gain that can enhance both communication reliability and sensing sensitivity. Third, the fine spatial granularity of FAS enables improved electromagnetic sampling of the environment, especially beneficial for sensing tasks in which small spatial displacements induce meaningful variations in reflected signals, as demonstrated in \cite{wang2024fas,ghadi2025isac,zou2024isac,zhou2024isac,tang2025isac}. 

\vspace{-2mm}
\subsection{State-of-the-Art}
In recent years, the information-theoretic and AI-enabled aspects of ISAC and FAS have been separately studied. For instance, the exact capacity-distortion region of single-receiver Gaussian ISAC was established in \cite{AhmadipourJoint2024}, and later extended to collaborative multiuser networks in \cite{AhmadipourCollaborative2023}. Finite blocklength ISAC performance was analyzed in \cite{NikbakhtFinite2024}, while \cite{JoudehCaire2024} introduced log-loss sensing mutual information (MI) as a unified metric for joint communication-sensing system design. On the other hand, the rate Cram\'{e}r-Rao bound (CRB) region for multiple-input multiple-output (MIMO) ISAC was investigated in \cite{hua2023mimo}, while \cite{an2023fund,ren2024fund} analyzed rate detection and rate CRB trade-offs in dual-functional radar/communication systems. Besides, surveys such as \cite{zhang2025int} and \cite{ald20256g} have documented the growing role of deep learning in waveform design, target estimation, and adaptive beamforming. Additionally, learning-based ISAC architectures including federated and fog-enabled implementations \cite{liu2024ai} and deep joint waveform beamforming designs \cite{vaezi2025ai} demonstrate significant empirical gains. Also, the information bottleneck framework \cite{Tishby2000} and its deep variational extensions \cite{Alemi} provide mechanisms to regulate the MI carried by latent representations. Deep task quantization and bit-limited radar and communication have also been explored in \cite{Shlezinger2021,xi2021bit}. 

Additionally, recent works have investigated ISAC systems under practical impairments and also in distributed and cooperative settings. In addition to the early results revealing the great benefits of FAS for ISAC in \cite{wang2024fas,ghadi2025isac,zou2024isac,zhou2024isac,tang2025isac}, recent attempts also look into the use of movable antennas for ISAC. For example, learning and optimization methods have been proposed to address challenges such as carrier frequency offset, time synchronization errors, and imperfect channel state information (CSI) in movable antenna assisted ISAC systems \cite{M1, M2, M3}. Moreover, \cite{M4} proposed a joint optimization framework for antenna positioning and beamforming to enhance sensing performance under communication constraints. In \cite{M5}, a unified MI framework was developed, where communication and sensing performance are jointly optimized via beamforming and antenna position design. Furthermore, \cite{M6} investigated ISAC under security constraints and proposed a learning-based approach for jointly optimizing beamforming and antenna movement under imperfect CSI. 

Meanwhile, the information-theoretic limits of MIMO-FAS were recently characterized in \cite{New2024MIMOFAS}, demonstrating superior diversity-multiplexing behavior and improved $q$-outage capacity relative to the classical MIMO system. FAS-assisted dirty multiple access channels were analyzed in \cite{Ghadi2024DMAC}, where spatial correlation was captured accurately using copula theory and sizable outage gains were reported. In the context of multiuser information-theoretic analysis, fluid antenna multiple access (FAMA) has also been investigated under strong interference channels in \cite{Ghadi2025fama}, where combining FAS-enabled spatial reconfigurability with simultaneous non-unique decoding has been shown to markedly improve the system performance. 

Several learning-based approaches such as deep reinforcement learning for ISAC precoding with planar FAS \cite{Wang2024ISACDRL}, deep learning-enabled FAMA \cite{Waqar2023FAMA}, and fast port selection methods \cite{Chai2022PortSelection}, demonstrate the synergy between FAS reconfigurability and data-driven design. More recently, opportunistic multiuser FAMA via team-theoretic reinforcement learning was further developed in \cite{Waqar2024OFAMA}, and active sensing-based beam alignment for MIMO-FAS was also proposed in \cite{Jiang2025BeamAlign}. 
 
\vspace{-2mm}
\subsection{Motivation and Contributions}
{\color{black}
The aforesaid studies demonstrate the effectiveness of combining reconfigurable FAS architectures with advanced signal processing and machine learning techniques to improve system robustness and performance under practical constraints. Despite the reported advantages of both ISAC and FAS, the interplay between FAS-assisted ISAC systems and representation-constrained models motivated by AI processing is not understood. While it is clear that FAS has the DoF to improve the post-encoder channel gain and mitigate the additional physical channel degradation observed after representation-constrained encoding, the communication and sensing performance will be ultimately limited by the representation-limited bounds.

Existing studies on representation-constrained ISAC systems primarily focus on fixed-antenna architectures, and do not capture the spatial sampling and selection dynamics introduced by FAS. Likewise, the current FAS literature focuses largely upon communication-centric performance metrics such as outage probability and diversity gain, without incorporating representation information constraints and sensing accuracy. Consequently, the fundamental limits of ISAC systems that operate under both representation information constraints and FAS-based port selection remain unknown.

 Motivated by this, this work provides the first information-theoretic characterization of FAS-assisted ISAC systems that operate under a pre-channel representation information constraint. Specifically, we adopt an information-theoretic perspective and model the encoder as a general stochastic mapping subject to an MI constraint, which can be interpreted as an abstraction of representation-limited processing modules, including learning-based encoders. Our focus is on the interaction between two consecutive system stages: a pre-channel representation-constrained mapping $X\mapsto Z$ and a post-encoder FAS-assisted bistatic ISAC channel. The variational information bottleneck (VIB) encoder presented later is not an additional physical system component, but an illustrative trainable realization of the stochastic mapping $p(z|x)$. The FAS receiver operates only on the physical channel outputs and therefore cannot modify the information retained in $Z$.
}

Our main contributions are summarized as follows.
\begin{itemize}
\item \underline{\textit{Representation-constrained}  FAS-assisted ISAC}. First, we develop a unified model where each ideal transmit symbol is compressed into a latent representation of limited information content, while the FAS receiver selects a single port according to the communication channel, with the selected port subsequently used for both communication and sensing. This structure induces the Markov relationship between the ideal symbol, its stochastic representation, and the received communication and sensing observations, and it embeds the port selection process via the selected channel gains.
\item \underline{\textit{Capacity-distortion characterization}}.
For Gaussian representations, we show that the representation bottleneck with a FAS receiver manifests as an additional effective noise term that degrades both communication and sensing performance at the selected port. Using this equivalence, we derive explicit expressions for the communication rate and sensing mean square distortion as functions of the FAS-selected channel gains. Furthermore, we establish matching converse and achievability bounds, thereby providing a complete description of the ISAC capacity-distortion region under the Rayleigh fading model.
\item \underline{\textit{FAS spatial DoF}}. Under isotropic rich scattering with the Jakes' correlation model, we furthermore characterize the effective spatial DoF of the FAS in the proposed representation-constrained ISAC system and show that the port selection gain is fundamentally limited by the physical length of the FAS region. The resulting diversity order equals the numerical rank of the Jakes correlation matrix and increases only with the FAS length, irrespective of the number of available ports.
\item \underline{\textit{Representation-limited performance ceiling:}} Finally, we illustrate that increasing the physical FAS aperture improves the selected post-encoder channel gain and allows the communication performance to approach its representation-limited ceiling. The same conclusion applies to sensing when the selected sensing gain also increases with the FAS aperture, as under the adopted common spatial fading benchmark. Importantly, this spatial gain does not modify $I(X;Z)$ or recover information discarded by the representation encoder.
\end{itemize}

{\color{black}
Note that the information-theoretic model developed in this work is hardware agnostic: rather than optimizing a particular material or antenna implementation, it represents each active radiating state through its corresponding communication and sensing channel coefficients. Hence, the analysis assumes comparable radiation characteristics across the candidate ports and ideal reconfiguration, while implementation-dependent losses and switching constraints are not explicitly modeled. As a result, our analysis is tied to the discrete port selection architecture, but not to a particular implementation.
}

\subsection{Organization and Mathematical Notations}
The remainder of this paper is organized as follows. Section \ref{sec:sys} introduces the system model, including the pre-channel representation-constrained bottleneck, the FAS architecture, and the adopted performance metrics. In Section \ref{sec:cap}, we derive the capacity-distortion region of the proposed representation-constrained FAS-enabled ISAC system and establish matching converse and achievability bounds. Section \ref{sec:stat} characterizes the statistics of the FAS selected port gains under spatially correlated fading and analyzes the finite port performance. Section \ref{sec:asy} investigates the asymptotic behavior with respect to the physical length of the fluid antenna, revealing the role of spatial DoF in reducing the post-encoder physical channel gap to the representation-limited ceiling. Section \ref{sec:var} presents a practical VIB encoder design that provides an illustrative realization of the theoretical representation model. Section \ref{sec:num} provides numerical results that validate the analysis and illustrate key trade-offs. Finally, Section \ref{sec:con} concludes the paper and discusses future research directions.

Throughout the paper, scalar variables are denoted by lowercase letters $x$, vectors by boldface lowercase letters $\mathbf{x}$, and matrices by boldface uppercase letters $\mathbf{X}$. The superscripts $(\cdot)^{\mathsf{T}}$ and $(\cdot)^{\dagger}$ denote transpose and Hermitian transpose, respectively. The operators $\mathbb{E}[\cdot]$ and $\mathbb{P}(\cdot)$ denote statistical expectation and probability. The notation $\mathcal{CN}(0,\sigma^2)$ represents a circularly symmetric complex Gaussian random variable with zero mean and variance $\sigma^2$. Also, MI is measured in bits and denoted by $I(\cdot;\cdot)$. The notation $\arg\max$ returns the maximizing index. The symbol $o(\cdot)$ follows standard asymptotic notation. Unless otherwise stated, $\log(\cdot)$ represents the base-$2$ logarithm. For convenience, the main symbols and variables used throughout the paper are summarized in Table~\ref{tab:notation}.

\begin{table}[t]
\caption{Summary of Notations}\label{tab:notation}
	\centering
	\renewcommand{\arraystretch}{1.2}
	\begin{tabular}{cl}
		\toprule
		\textbf{Symbol} & \textbf{Description} \\
		\midrule
		$M$ & Transmit information message \\
		$X$ & Ideal transmitted symbol  \\
		$Z$ & Latent representation \\
		$W_z$ & Representation noise \\
		$C_{\mathrm{rep}}$ & Prescribed representation budget\\
		$N_z$ & Representation noise variance \\
		$N_z^\star$ & Minimum representation noise variance permitted by
		$C_{\mathrm{rep}}$ \\
		$P$ & Reference useful symbol power, $\mathbb E[|X|^2]$ \\
		$N_0$ & Thermal noise variance \\
		$Y_c$ & Communication received signal \\
		$Y_s$ & Processed sensing sufficient statistic\\
		$\theta$ & Scalar complex target response coefficient \\
		$\sigma_\theta^2$ & Variance of the sensing parameter \\
		$h_{c,\ell}$ & Communication channel gain at port $\ell$ \\
		$h_{s,\ell}$ & Sensing channel gain at port $\ell$ \\
		$\mathbf{h}_c$ & Communication channel vector \\
		$\mathbf{h}_s$ & Sensing channel vector \\
		$\ell^\star$ & Selected FAS port index \\
		$U(\cdot)$ & Utility function used for port selection \\
		$L$ & Number of FAS ports \\
		$W$ & Physical length of the fluid antenna\\
		$d$ & Inter-port spacing in the FAS \\
		$\lambda$ & Carrier wavelength \\
		$\mathbf{R}$ & Spatial correlation matrix of the FAS channels \\
		$\gamma_c^\star$ & Selected port communication gain \\
			$\gamma_s^\star$ & Selected port  sensing gain \\
		$\Gamma_c^\star$ & Effective communication SNR \\
				$\Gamma_s^\star$ & Effective sensing SNR\\
		$\Gamma_{\mathrm{rep}}$ & Representation-imposed SNR ceiling \\
		$L'(W)$ & Effective number of spatial DoF \\
		$R$ & Achievable communication rate \\
		$D_s$ & Sensing distortion \\
		\bottomrule
	\end{tabular}
\end{table}


\section{System Model}\label{sec:sys}

{\color{black}
As illustrated in Fig.~\ref{fig:system}, we consider a single-target bistatic ISAC system comprising a transmitter, a spatially separated FAS-equipped receiver, and one target. The transmitter conveys a data message to the receiver over the direct communication path, while the same transmitted waveform illuminates the target and produces a reflected sensing return over the transmitter-target-receiver path. The FAS receiver jointly decodes the communication message and estimates a scalar target response parameter $\theta$ using one of $L$ candidate ports. The communication and sensing observations represent the sufficient statistics obtained after separation of the direct and reflected components, for example through delay-domain processing and matched filtering. The message $M$ is encoded into a length-$n$ ideal Gaussian codeword $\{X_t\}_{t=1}^{n}$ with independent and identically distributed (i.i.d.) components $X_t\sim\mathcal{CN}(0,P)$, which provides an information-theoretic abstraction for characterizing fundamental performance limits. The actual waveform entering the physical channel is the latent representation $Z_t$ produced by the representation-constrained stochastic mapping $p(z|x)$. This mapping is subject to a finite representation budget that limits the MI between the ideal symbol and its latent representation, rather than the capacity of the physical communication channel.
}

\begin{figure*}[t]
\centering
\includegraphics[width=1.9\columnwidth]{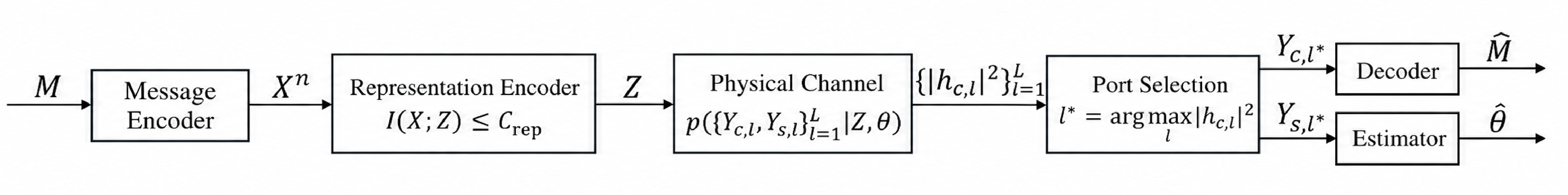}
\caption{Information flow abstraction of the representation-constrained single target bistatic FAS-aided ISAC system.}\label{fig:system}
\end{figure*}

\subsection{Pre-Channel  Representation Bottleneck}
The representation-constrained encoder implements a memoryless stochastic mapping from the ideal symbol $X$ to a latent representation $Z$ under a prescribed representation budget. Formally, the mapping is described by a conditional distribution $p(z|x)$ that must satisfy an MI constraint
\begin{align}\label{eq:AI_constraint}
I(X;Z) \le C_{\mathrm{rep}},
\end{align}
{\color{black}
where $C_{\mathrm{rep}}$ is a prescribed information budget for the latent representation, measured in bits per complex input symbol. It quantifies the maximum amount of information about the ideal symbol $X$ that may be retained in the pre-channel representation $Z$. The value of $C_{\mathrm{rep}}$ is selected as a system operating target and is not a universal capacity of AI models. Whether a particular encoder can satisfy this target, and the task performance obtained at that target, depend on factors such as network architecture, latent dimension, numerical precision, prior distribution, and optimization procedure. Different encoder architectures may thus be trained under the same value of $C_{\mathrm{rep}}$ while achieving different task performance.
}

Following standard information bottleneck formulations, we adopt a Gaussian representation model 
\begin{equation}\label{eq:test_channel}
Z = X + W_z,
\end{equation}
where $W_z \sim \mathcal{CN}(0,N_z)$ with $W_z$ independent of $X$, where $W_z$ is a zero-mean circularly symmetric complex Gaussian representation noise term and $N_z$ is its variance. The Gaussian representation noise assumption is consistent with classical rate-distortion theory, where the optimal test-channel under quadratic distortion is Gaussian \cite{CoverThomas, berger}. It is also widely adopted in information bottleneck formulations and their variational implementations \cite{Tishby2000, Alemi}, providing a tractable and meaningful approximation of representation information constraints in learning systems. For this model, the MI satisfies
\begin{align}
I(X;Z) = \log_2\!\left(1 + \frac{P}{N_z}\right),
\end{align}
where $P=\mathbb E[|X|^2]$ is the reference useful symbol power in the adopted Gaussian test channel abstraction.  {\color{black}For the adopted Gaussian test channel, the representation budget has the direct operational interpretation 
$
		\frac{N_z}{P}
		=
		\frac{1}{2^{C_{\mathrm{rep}}}-1}
$. 
Hence, $C_{\mathrm{rep}}$ uniquely determines the normalized representation noise level and the corresponding representation-limited signal-to-noise ratio (SNR) ceiling. Importantly, the FAS operates after the mapping $X\mapsto Z$ and does not alter $I(X;Z)$; it only affects the quality of the subsequent physical channel observation.
}

{\color{black} Since the achievable communication and sensing performance improves as $N_z$ decreases, the boundary of the representation information constraint is active at the optimal operating point. Therefore, we have}
\begin{align}\label{eq:Nz_star}
N_z^\star = \frac{P}{2^{C_{\mathrm{rep}}}-1},
\end{align}
which characterizes the effective information loss by the representation bottleneck. Thus, the pre-channel representation-constrained processing manifests as an additive distortion whose impact depends on the propagation channel. {\color{black}Additionally, the Gaussian test channel is an equivalent information-theoretic abstraction rather than a literal RF stage that affects the transmit power. Here, $P$ denotes the reference useful signal variance, while $N_z^\star/P$ quantifies the normalized uncertainty introduced by the representation mapping.
}

\subsection{FAS-Equipped Receiver Model}
{\color{black}
Throughout this paper, and without loss of generality, we assume a $1$D FAS discrete port receiver comprising $L$ predefined candidate ports and a single RF chain. The receiver includes a channel sounding and control unit together with a port selection network. During each channel coherence interval, the candidate ports, or a selected subset of them, are probed, the prescribed selection metric is evaluated, and the selected port is connected to the RF chain. Accordingly, the reconfiguration variable is the selected port index $\ell^\star$, rather than the continuously optimized physical coordinate of a movable element, and the $L$ candidate ports are not simultaneously active.} The FAS architecture affects the system exclusively through the spatial structure of the received channel realizations and the subsequent port selection rule, while the underlying physical ISAC channel law remains unchanged. The ISAC receiver, which jointly decodes the transmitted data message and estimates the sensing parameter, is equipped with a fluid antenna exposing a set of $L$ ports indexed by $\ell \in \{1,\dots,L\}$. These ports correspond to distinct physical locations, enabling the ISAC receiver to exploit small-scale spatial variations in the propagation environment for both communication and sensing tasks. At each channel use, the receiver selects a single port according to a deterministic selection rule
\begin{align}\label{eq:selection_rule}
\ell^\star = \arg\max_{\ell} U\bigl(h_{c,\ell}, h_{s,\ell}\bigr),
\end{align}
in which $h_{c,\ell}$ and $h_{s,\ell}$ denote the communication and sensing channel gains at port $\ell$, respectively, and $U(\cdot)$ denotes the utility function used for port selection. Here, for simplicity, we have taken $U = |h_{c,\ell}|^2$. 
{\color{black}
\begin{remark}[Port acquisition and selection]
Practical FAS receivers may employ different port acquisition and
selection strategies, facing a trade-off between spatial diversity
acquisition and channel estimation overhead. Exhaustive instantaneous selection probes all candidate ports and selects the port maximizing the prescribed utility. Reduced-overhead alternatives probe only a subset of ports and infer the remaining channel responses using spatial or temporal correlation, whereas statistical selection relies on long-term channel information. In ISAC systems, the selection utility may additionally combine communication and sensing objectives. Instantaneous port acquisition generally introduces a rate penalty associated with channel training, together with practical costs such as scanning time and switching latency. To isolate the fundamental effects of the FAS aperture, spatial correlation, and representation budget $C_{\mathrm{rep}}$ from a specific channel estimation protocol or hardware realization, we adopt ideal exhaustive selection with perfect instantaneous communication channel knowledge and zero acquisition overhead, i.e., $U(h_{c,\ell},h_{s,\ell})=|h_{c,\ell}|^2
$. This communication-centric policy provides an idealized upper
benchmark for the communication performance under the adopted selection rule, but it is not generally optimal for sensing or for the joint communication-sensing trade-off. The impact of this port-selection policy on the sensing gain will be later investigated, using a spatial correlation similarity metric between the sensing and communication channels.
\end{remark}}
Therefore, given the latent channel input $Z$, the per-port received signals are given by
\begin{align}
Y_{c,\ell} = h_{c,\ell} Z + N_c,
\end{align}
where the noise term $N_c$ is an independent circularly symmetric complex Gaussian random variable with variance $N_0$.

{\color{black}
To distinguish the propagation geometry from the physical geometry of the FAS aperture, we decompose the channel coefficients as
\begin{align}
h_{q,\ell}=\sqrt{\beta_q(r_q)}\,g_{q,\ell},~q\in\{c,s\},
\end{align}
where $\beta_q(r_q)$ represents the distance large-scale channel gain and $g_{q,\ell}$ denotes the normalized small-scale fading coefficient. A conventional power law model can be written as $\beta_q(r_q)=\beta_{0,q}(r_q/r_0)^{-\alpha_q}$, where $r_q$ denotes the corresponding propagation distance, $r_0$ is a reference distance, and $\alpha_q$ is the path-loss exponent. For a two-leg sensing path, $\beta_s$ can equivalently incorporate the product of the transmitter-target and target-receiver path-loss factors. The subsequent analysis remains valid for arbitrary values of $\beta_c$ and $\beta_s$ because they are included in the channel coefficients.
}

Now, substituting the representation model \eqref{eq:test_channel} into the communication observation yields
\begin{align}
	Y_{c,\ell}
	&=
	h_{c,\ell}X
	+
	\underbrace{h_{c,\ell}W_z+N_c}_{\widehat N_{c,\ell}},
\end{align}
where the equivalent communication noise remains Gaussian with
variance
\begin{align}
	\sigma_{c,\ell}^{2}
	=
	N_0+|h_{c,\ell}|^{2}N_z^\star.
\end{align}
Consequently, the effective communication SNR at port $\ell$ is
\begin{align}
	\Gamma_{c,\ell}
	=
	\frac{|h_{c,\ell}|^{2}P}
	{N_0+|h_{c,\ell}|^{2}N_z^\star}.
\end{align}

{\color{black}
	For sensing, we adopt a normalized equivalent linear Gaussian
	sufficient statistic model obtained after separation of the bistatic
	reflected component and coherent matched filtering:
	\begin{align}
		Y_{s,\ell}
		=
		\sqrt{\frac{|h_{s,\ell}|^{2}P}{\sigma_\theta^{2}}}\,\theta
		+
		\widetilde N_{s,\ell},
		\label{eq:sensing_sufficient_statistic}
	\end{align}
	where $\theta\sim\mathcal{CN}(0,\sigma_\theta^{2})$ denotes the scalar
	complex target-response coefficient and $h_{s,\ell}$ is the composite
	bistatic propagation coefficient excluding $\theta$. The above model
	is an equivalent information-theoretic sensing abstraction rather
	than a waveform level radar equation.
	
	The representation uncertainty contributes to the equivalent
	sensing noise variance according to
	\begin{align}
		\widetilde N_{s,\ell}
		\sim
		\mathcal{CN}\left(
		0,N_0+|h_{s,\ell}|^{2}N_z^\star
		\right),
	\end{align}
	which yields
	\begin{align}
		\Gamma_{s,\ell}
		=
		\frac{|h_{s,\ell}|^{2}P}
		{N_0+|h_{s,\ell}|^{2}N_z^\star}.
	\end{align}
}

Under the port selection rule \eqref{eq:selection_rule}, the system
operates with the selected communication and sensing gains
$
	H_c^\star=h_{c,\ell^\star}$,  
	$H_s^\star=h_{s,\ell^\star}$ 
and the corresponding effective SNRs
$\Gamma_c^\star$ and $\Gamma_s^\star$.

{\color{black}
The respective roles of the representation encoder and the FAS are thus distinct. The representation mapping determines $N_z^\star$ and hence the ceiling $\Gamma_{\mathrm{rep}}=P/N_z^\star$, whereas the FAS changes only the selected physical channel gain. In particular,
	\begin{align}
		\Gamma_{\mathrm{rep}}-\Gamma_q(\gamma)
		=
		\frac{PN_0}
		{N_z^\star\left(N_0+\gamma N_z^\star\right)},
	\end{align}
which decreases monotonically with $\gamma$. As a consequence, FAS port selection reduces the post-encoder channel gap to the representation-limited ceiling but does not recover information discarded by the encoder.
}

%

\subsection{Performance Metrics}
ISAC performance is characterized in terms of both communication reliability and sensing accuracy. Let $\mathcal{M}$ be the message set and $\hat{M}$ be the message estimate produced by the communication receiver after observing $Y_{c,\ell^\star}^n$ over $n$ channel uses. Therefore, the communication rate is defined as
\begin{equation}
R = \frac{1}{n}\log_2|\mathcal{M}|,
\end{equation}
and a rate $R$ is said to be achievable if the decoding error probability satisfies $\mathbb{P}(\hat{M}\neq M)\to 0$ as $n\to\infty$.

For sensing, the receiver forms an estimate $\hat{\theta}$ of the underlying parameter $\theta$ using the sequence of sensing observations $Y_{s,\ell^\star}^n$. We measure sensing performance by the mean square error (MSE), defined as
\begin{equation}
D_s = \mathbb{E}\!\left[\,\left|\theta - \hat{\theta}(Y_{s,\ell^\star}^n)\right|^2\,\right], 
\end{equation}
where a distortion level $D_s$ is achievable if there exists an estimator such that the above MSE is attained in the limit of large blocklength. 
{\color{black}
In this work, $\theta$ is a scalar complex target response coefficient associated with the single target in the bistatic sensing link. This choice yields the linear Gaussian sufficient statistic in \eqref{eq:sensing_sufficient_statistic} and permits an exact Bayesian minimum MSE (MMSE) characterization. More general parameters such as delay or angle can be represented through a local linearized or vector extension of the sensing model, but their waveform Fisher information analysis is beyond the scope of our work.
}

Note that the fundamental trade-off between communication and sensing under the representation information constraint is captured by the capacity-distortion region, denoted as $\mathcal{R}_{\mathrm{rep\mbox{-}FAS}}(C_{\mathrm{rep}})$, which consists of all pairs $(R,D_s)$ that are simultaneously achievable when the input representation mapping is constrained to meet $I(X;Z)\le C_{\mathrm{rep}}$. This region characterizes the performance limits of ISAC with a representation-constrained transmitter and an FAS-enabled receiver.

\section{Capacity-Distortion Region}\label{sec:cap}
This section establishes the fundamental communication-sensing trade-off for ISAC with a representation-constrained transmitter and an FAS-enabled receiver. The analysis captures, in closed-form, how the  representation bottleneck and the selected port gains jointly determine the maximum achievable communication rate and the minimum attainable sensing distortion. The result applies to general FAS fading statistics and any deterministic port selection rule.

\begin{theorem}[Capacity-Distortion Region]\label{thm:region}
Consider an ISAC system with the representation information constraint $I(X;Z)\le C_{\mathrm{rep}}$, employing the Gaussian representation model \eqref{eq:test_channel}, and having an FAS receiver using the selection rule in \eqref{eq:selection_rule}. The achievable communication rate $R$ and sensing distortion $D_s$ satisfy
\begin{align}\label{eq:R_region}
R &\le \mathbb{E}\left[\log_2\left(1 + \frac{|H_c^\star|^2 P}{\,N_0 + |H_c^\star|^2 N_z^\star}\right)\right]
\end{align}
and
\begin{align}\label{eq:D_region}
D_s &\ge \mathbb{E}\left[\frac{\sigma_\theta^2}{1 + \frac{|H_s^\star|^2 P}{\,N_0 + |H_s^\star|^2 N_z^\star}}\right],
\end{align}
where $N_z^\star$ denotes the minimum representation noise variance permitted by the representation bottleneck, and the expectation is taken over the distribution of the selected port gains $(H_c^\star, H_s^\star)$. Furthermore, these bounds are tight such that every pair $(R,D_s)$ satisfying \eqref{eq:R_region}-\eqref{eq:D_region} is achievable.
\end{theorem}

\begin{IEEEproof}[Proof]
The converse follows from Fano's inequality and the Markov chain
$
M\to X^n\to Z^n
\to (Y_c^n,Y_s^n)
$. 
Conditioning on the selected port index gives
\begin{align}
	nR
	&\leq
	I\!\left(M;Y_{c,\ell^\star}^n,\ell^\star\right)
	+n\epsilon_n\\
	&\leq
	I\!\left(X^n;Y_{c,\ell^\star}^n\mid\ell^\star\right)
	+n\epsilon_n,
\end{align}
where $\epsilon_n\to0$. Memorylessness and the Gaussian test channel
model then yield the corresponding single-letter quantity
$I(X;Y_c\mid \ell^\star)$, which gives \eqref{eq:R_region}. 


For sensing, under the scalar linear Gaussian model in
\eqref{eq:sensing_sufficient_statistic}, the conditional Bayesian
MMSE at the selected port is
$
\operatorname{mmse}
\left(\theta\mid Y_{s,\ell^\star},H_s^\star\right)
=
\frac{\sigma_\theta^2}{1+\Gamma_s^\star}.
$
Averaging over the selected sensing gain yields
\eqref{eq:D_region}.

Achievability is established using Gaussian codebooks for $X^n$, the optimal Gaussian mapping $X^n \mapsto Z^n$, and standard maximum likelihood and  minimum MSE estimators at the selected port. The resulting communication rate and sensing distortion attain the expressions on the right-hand sides of \eqref{eq:R_region} and \eqref{eq:D_region}, completing the characterization.
\end{IEEEproof}

Theorem~\ref{thm:region} holds for arbitrary joint fading distributions of the port dependent channel gains $\{(h_{c,\ell},h_{s,\ell})\}$ and for any deterministic utility selection policy. In the following sections, we investigate the implications of this general region for commonly studied FAS models and characterize how the number of ports and the underlying fading statistics shape the achievable ISAC performance.

\section{FAS Statistics and Finite-$L$ Performance}\label{sec:stat}
The capacity-distortion region in Theorem~\ref{thm:region} depends on the statistics of the selected port gains $H_c^\star = h_{c,\ell^\star}$ and $H_s^\star = h_{s,\ell^\star}$ through $\gamma_c^\star = |H_c^\star|^2$ and $\gamma_s^\star = |H_s^\star|^2$. These gains are strongly affected by the spatial correlation among ports, where such correlation is typically characterized by Jakes' model under rich scattering. This section explains how existing analytical results for FAS selected port gains can be combined with the proposed representation-constrained ISAC framework.

{\color{black}
	For the considered $1$D FAS, the $\ell$th port is located
	at
		$x_\ell=(\ell-1)d$ with 
		$d=\frac{W}{L-1}$, where $W$ is the physical aperture length and $d$ is the uniform
	inter-port spacing. Hence, the physical separation between ports
	$k$ and $\ell$ is
	\begin{align}
		\Delta_{k,\ell}=|x_k-x_\ell|
		=\frac{|k-\ell|W}{L-1}.
	\end{align}
	
	Under narrowband far-field propagation with isotropic rich scattering
	and no dominant line-of-sight component, the normalized small-scale
	channel vector is circularly symmetric complex Gaussian. Consequently,
	the communication channel $\mathbf{h}_c = (h_{c,1},\dots,h_{c,L})^{\mathsf{T}}$ satisfies
	\begin{align}
		\mathbf h_c\sim\mathcal{CN}
		\left(\mathbf 0,\beta_c\mathbf R\right),
	\end{align}
	where the normalized Jakes correlation matrix is
	\begin{align}
		[\mathbf R]_{k,\ell}
		=J_0\left(\frac{2\pi}{\lambda}\Delta_{k,\ell}\right).
	\end{align}
The same marginal model is adopted for the sensing channel $\mathbf{h}_s = (h_{s,1},\dots,h_{s,L})^{\mathsf{T}}$, with large-scale gain $\beta_s$. The Rayleigh Jakes model is used to obtain explicit FAS statistics and spatial DoF results. However, it is worth mentioning that the general capacity-distortion characterization in Theorem~\ref{thm:region} remains applicable to other channel distributions.
}

{\color{black}
To capture the potential similarity between the communication and bistatic sensing gains, we adopt the correlated channel model
\begin{align}
\mathbf h_s=\rho_{cs}\sqrt{\frac{\beta_s}{\beta_c}}\,\mathbf h_c+\sqrt{1-\rho_{cs}^2}\,\mathbf v_s,
\end{align}
in which $\mathbf v_s\sim\mathcal{CN}(\mathbf 0,\beta_s\mathbf R)$ is independent of $\mathbf h_c$ and $0\leq\rho_{cs}\leq1$ can be interpreted as the communication-sensing spatial correlation coefficient. The case $\rho_{cs}=0$ corresponds to independent spatial channels, whereas $\rho_{cs}=1$ represents the idealized common spatial fading benchmark, for which the communication selection gain is fully transferred to sensing.
}


Using the SNR maximizing selection rule\footnote{The port selection rule considered in this work is based on maximizing the communication channel gain, which provides an analytically tractable and physically meaningful benchmark. However, in general ISAC systems, one may consider more sophisticated selection strategies that jointly account for both communication and sensing objectives. For instance, multi-criterion selection rules and weighted utility functions could be employed to balance rate and sensing performance, particularly when their priorities differ. The analysis of such joint communication--sensing selection policies is beyond the scope of this work and is left for future research.}
\begin{equation}\label{eq:FAS_max_rule}
\ell^\star = \arg\max_{\ell\in\{1,\dots,L\}} |h_{c,\ell}|^2,
\end{equation}
the selected port communication gain is defined as $\gamma_c^\star = \max_\ell |h_{c,\ell}|^2$, and the sensing gain at the selected port is given by $\gamma_s^\star = |h_{s,\ell^\star}|^2$, where the corresponding distributions are derived in terms of the spatial correlation matrix $\mathbf{R}$ in \cite{new2024fluid}. {\color{black}
The distribution of $\gamma_s^\star$ depends on $\rho_{cs}$. For $\rho_{cs}=0$, the communication selected index is independent of the sensing channel, and therefore $\gamma_s^\star$ retains the marginal sensing gain distribution. For $\rho_{cs}>0$, communication-centric port selection also induces a sensing selection gain, whose strength increases with the communication-sensing spatial correlation.} It has been proven in \cite{new2024fluid} that there exists an integer $L'(W)\leq L$ depending only on the length $W$ (and thus on $\mathbf{R}$) such that increasing the number of ports beyond $L'(W)$ does not yield additional diversity gain. In other words, the effective number of diversity branches of the FAS is governed by the physical length, and not by the nominal port count.

	
In our representation-constrained ISAC setting, the representation bottleneck introduces an additional noise term $N_z^\star$ in the effective SNRs at the selected FAS port. Specifically, for a given realization of the selected gain $\gamma_c^\star$, the effective communication SNR is denoted as
\begin{equation}\label{eq-snrc}
\Gamma_c^\star=\frac{\gamma_c^\star P}{N_0 + \gamma_c^\star N_z^\star}, 
\end{equation}
and the sensing SNR is defined as
\begin{equation}\label{eq-snrs}
\Gamma_s^\star=\frac{\gamma_s^\star P}{N_0 + \gamma_s^\star N_z^\star}.
\end{equation}
Therefore, the capacity-distortion region in Theorem~\ref{thm:region} can then be written as
\begin{align}
R&=\mathbb{E}\left[\log_2\!\left(1+\Gamma_c^\star\right)\right]\notag\\
&=\int_0^\infty\log_2\left(1+\frac{xP}{N_0 + x N_z^\star}\right)f_{\gamma_c^\star}(x)\,dx\label{eq:R_integral_FAS}
\end{align}
and
\begin{align}
D_s&=\mathbb{E}\left[\frac{\sigma_\theta^2}{1+\Gamma_s^\star}\right]\notag\\
&=\int_0^\infty\frac{\sigma_\theta^2}{1 + \frac{xP}{N_0 + x N_z^\star}}f_{\gamma_s^\star}(x)\,dx,\label{eq:D_integral_FAS}
\end{align}
where $f_{\gamma_c^\star}(\cdot)$ is the probability density function (PDF) of the FAS selected gain obtained in \cite{new2024fluid}, and $f_{\gamma_s^\star}(\cdot)$ is determined by the sensing side correlation model.

Equations \eqref{eq:R_integral_FAS} and \eqref{eq:D_integral_FAS} provide a direct bridge between the representation-constrained ISAC framework and the existing FAS theory such that the representation bottleneck appears only via the additional noise variance $N_z^\star$, while the finite-$L$ behavior and the influence of the length $W$ enter through the FAS-selected gain distributions already characterized in \cite{new2024fluid}.

\section{Asymptotic Behavior: FAS length, Spatial DoF, and representation information constraints}\label{sec:asy}
Here, we analyze how the physical length of a fluid antenna influences the maximum gain extracted through port selection and how this, in turn, interacts with the representation-induced noise $N_z^\star$.  In contrast to selection schemes over i.i.d.~fading, the selected port gain in an FAS is fundamentally limited by the spatial DoF of a $1$D region of length $W$.

\subsection{Spatial DoF and Selected Port Gain Saturation}
Under the Jakes' spatial correlation model, the communication channel vector $\mathbf{h}_c = [h_{c,1},\dots,h_{c,L}]^{\mathsf{T}}$ is complex Gaussian with Toeplitz covariance
\begin{align}
R_{k,\ell}(W)= J_0\left(\frac{2\pi}{\lambda}\frac{W}{L-1}\,|k-\ell|\right).
\end{align}
As shown in \cite{new2024fluid}, the high-SNR outage probability of the FAS selected channel satisfies
\begin{align}
\mathbb{P}\left(|h_{c,\mathrm{FAS}}| < \Omega\right)= \frac{1}{\det(\mathbf{R})}\,\Omega^{2L'(W)} + o(\Omega^{2L'(W)}),
\end{align}
where $L'(W)$ is the numerical rank of the correlation matrix $\mathbf{R}$ and represents the effective number of spatial DoF supported by an FAS of length $W$. Consequently, the diversity order is given by $D_{\mathrm{FAS}}(W) \approx L'(W)$, which is independent of the number of ports $L$ once $L$ is large enough. Thus, increasing $L$ does not cause the selected port gain $\gamma_c^\star = \max_{\ell}|h_{c,\ell}|^2$ to diverge. Instead, as $L\to\infty$ with fixed $W$, $\gamma_c^\star$ converges in distribution to a length-limited random variable with finite moments whose distribution is governed by $L'(W)$. Hence, unlike the i.i.d.~Rayleigh case where $\gamma_c^\star(L)=\Theta(\log L)$, a $1$D FAS exhibits length-limited gain scaling.

\subsection{Consequences for Constrained ISAC}
Recall the effective SNR defined in \eqref{eq-snrc}, where the representation bottleneck introduces the ceiling
\begin{align}\label{eq:theopt}
\Gamma_{\mathrm{rep}} = \frac{P}{N_z^\star} = 2^{C_{\mathrm{rep}}}-1.
\end{align}

{\color{black}
With a fixed finite aperture $W$, the dense port selected gain $\gamma_c^\star(W)$ remains finite almost surely under the adopted Jakes-correlated Gaussian field. Consequently, we have $\Gamma_c^\star(W)<\Gamma_{\mathrm{rep}}$ almost surely, and hence,
	\begin{align}
		R(W)
		=
		\mathbb E\!\left[\log_2\!\left(1+\Gamma_c^\star(W)\right)\right]
		<
		C_{\mathrm{rep}}.
	\end{align}
}

Larger FAS length increases $L'(W)$, improves the upper tail of
$\gamma_c^\star$, and raises the effective SNR toward the
representation-limited ceiling. If the selected communication gain
diverges in probability as $W\to\infty$, the effective SNR approaches
the representation-imposed ceiling. The corresponding conditional
limits are summarized in the following theorem.

\begin{theorem}[Length-Dependent Approach to the
	Representation-Limited Bounds]\label{thm:AI_FAS_length_compensation}
Let $\gamma_c^\star(W)$ denote the FAS-selected gain under Jakes correlated fading model. Then
\begin{enumerate}
\item For fixed FAS length $W < \infty$,
\begin{align}
R(W) &<C_{\mathrm{rep}}
\end{align}
and
\begin{align}
D_s(W) &>\frac{\sigma_\theta^2}{2^{C_{\mathrm{rep}}}}.
\end{align}
\item If
$\gamma_c^\star(W)\to\infty$ in probability, then
\begin{align}
	\lim_{W\rightarrow\infty}R(W)
	=
	C_{\mathrm{rep}}.
\end{align}
If, in addition, $\gamma_s^\star(W)\to\infty$ in probability, as in
the idealized common spatial fading benchmark, then
\begin{align}
	\lim_{W\rightarrow\infty}D_s(W)
	=
	\frac{\sigma_\theta^2}{2^{C_{\mathrm{rep}}}}.
\end{align}
\end{enumerate}
\end{theorem}


\begin{IEEEproof}
	For any fixed finite aperture $W$, the selected communication and
	sensing gains $\gamma_c^\star(W)$ and $\gamma_s^\star(W)$ are finite
	almost surely under the adopted Jakes correlated fading model
	\cite{new2024fluid}. Therefore, for $q\in\{c,s\}$,
$
		\Gamma_q^\star(W)
		=
		\frac{\gamma_q^\star(W)P}
		{N_0+\gamma_q^\star(W)N_z^\star}
		<
		\frac{P}{N_z^\star}
		=
		2^{C_{\mathrm{rep}}}-1
$
	almost surely. It follows that
$
		\log_2\!\left(1+\Gamma_c^\star(W)\right)
		<
		C_{\mathrm{rep}},
$
	and hence $R(W)<C_{\mathrm{rep}}$. Similarly,
$
		\frac{\sigma_\theta^2}
		{1+\Gamma_s^\star(W)}
		>
		\frac{\sigma_\theta^2}
		{2^{C_{\mathrm{rep}}}}
$, 
	which yields
	$D_s(W)>\sigma_\theta^2/2^{C_{\mathrm{rep}}}$.
	
	For Part~2, under the stated condition
	$\gamma_c^\star(W)\to\infty$ in probability,
	the corresponding effective SNR satisfies
$
		\Gamma_c^\star(W)
		\to
		\frac{P}{N_z^\star}
		=
		2^{C_{\mathrm{rep}}}-1
$
	in probability. Since
$
		0
		\leq
		\log_2\!\left(1+\Gamma_c^\star(W)\right)
		\leq
		C_{\mathrm{rep}}
$, 
	uniform boundedness implies convergence of the corresponding
	expectations, and thus
	$R(W)\to C_{\mathrm{rep}}$. Likewise, if
	$\gamma_s^\star(W)\to\infty$ in probability, 
$
		\frac{\sigma_\theta^2}
		{1+\Gamma_s^\star(W)}
		\to
		\frac{\sigma_\theta^2}
		{2^{C_{\mathrm{rep}}}}
$
	in probability. Since the distortion term is bounded between $0$
	and $\sigma_\theta^2$, its expectation also converges, yielding
$
		D_s(W)
		\to
		\frac{\sigma_\theta^2}
		{2^{C_{\mathrm{rep}}}}$. 
\end{IEEEproof}

{\color{black}
When $\rho_{cs}=0$ and port selection is based solely on the communication channel, the sensing gain retains its marginal distribution and does not increase with the FAS communication-selection gain. In this case, the sensing limit does not follow from an increase in the FAS aperture; approaching that limit requires either communication-sensing spatial correlation or a joint communication-sensing port-selection rule.
}

{\color{black}
These results reveal that the physical FAS length, rather than the nominal number of densely deployed ports, governs the spatial gain available to reduce the post-encoder physical-channel gap to the representation-limited ceiling. This conclusion applies directly to communication-centric port selection, and extends to sensing only when the selected sensing gain also increases with the FAS aperture.
}

\section{VIB Encoder Design}\label{sec:var}
The theoretical analysis in previous sections models the pre-channel representation mapping $X \mapsto Z$ as a Gaussian  channel that satisfies the information constraint $I(X;Z)\le C_{\mathrm{rep}}$. In this section, we describe a practical neural implementation of this mapping based on the VIB, which provides a Gaussian-parameterized stochastic encoder and can be trained directly from data.  Specifically, the VIB implementation is introduced as an illustrative realization of the information bottleneck model. 
{\color{black}
		This section establishes how the theoretical bottleneck model can be
		approximately realized through a trainable stochastic encoder and how
		the resulting representation can be calibrated against the Gaussian
		test channel benchmark. The VIB implementation is an illustrative
		learning realization of the prescribed representation
		constraint and is not assumed to exactly reproduce the additive
		Gaussian test channel.
	}
\subsection{VIB Objective}
 The bottleneck is positioned before the physical channel, so the latent representation $Z$ is generated from the ideal transmitted symbol $X$ under the information constraint $I(X;Z)\le C_{\mathrm{rep}}$. The physical ISAC channel then operates on $Z$, producing the communication and sensing observations $(Y_c,Y_s)$. Accordingly, the dependency structure follows the Markov chain $X\rightarrow Z \rightarrow (Y_c,Y_s)$.

%

{\color{black}
	For a prescribed representation  budget
	$C_{\mathrm{rep}}$, the practical encoder design is formulated as
	\begin{align}
		\begin{aligned}
			\underset{\phi,\psi}{\operatorname{minimize}}\quad&
			\mathcal{L}_{\mathrm{task}}(\phi,\psi)\\
			\operatorname{subject~to}\quad&
			\overline{I}_{\phi}(X;Z)\leq C_{\mathrm{rep}},
		\end{aligned}
		\label{eq:vib_constrained}
	\end{align}
	where
	\begin{align}
		\mathcal{L}_{\mathrm{task}}(\phi,\psi)
		=
		\mathbb{E}_{p(x,m)}
		\mathbb{E}_{p_{\phi}(z|x)}
		[-\log p_{\psi}(m|z)]
	\end{align}
	denotes the task loss. The quantity
	$\overline{I}_{\phi}(X;Z)$ is the standard variational upper
	bound on the encoder MI, given by
	\begin{align}
		I_{\phi}(X;Z)
		\leq
		\overline{I}_{\phi}(X;Z)
		\triangleq
		\mathbb{E}_{p(x)}
		D_{\mathrm{KL}}
		\left(
		p_{\phi}(z|x)\Vert p(z)
		\right).
		\label{eq:vib_rate_bound}
	\end{align}
	
	The associated Lagrangian is
	\begin{align}
		\mathcal{L}_{\mathrm{VIB}}
		=
		\mathcal{L}_{\mathrm{task}}
		+
		\beta
		\left(
		\overline{I}_{\phi}(X;Z)-C_{\mathrm{rep}}
		\right),
		\label{eq:vib_lagrangian}
	\end{align}
	where $\beta\geq0$ is the dual variable associated with the
	representation rate constraint. Thus, $C_{\mathrm{rep}}$ is the
	prescribed operating budget, whereas $\beta$ is adjusted during
	training to enforce that budget. In particular, $\beta$ is not
	itself a representation budget and does not have a universal
	one-to-one relationship with $C_{\mathrm{rep}}$.
}

\subsection{Gaussian Reparameterization}
To align the neural network encoder with the optimal  channel structure, we model
\begin{align}
p_\phi(z|x) = \mathcal{CN}\left(\mu_\phi(x),\Sigma_\phi(x)\right), 
\end{align}
and using the standard reparameterization technique, we write
\begin{equation}\label{eq:reparam}
z = \mu_\phi(x)+ \Sigma_\phi^{1/2}(x)\,\epsilon,~\mbox{with }\epsilon \sim \mathcal{CN}(0,I),
\end{equation}
which allows low variance gradient estimation.

Choosing the prior as $p(z)=\mathcal{CN}(0,\sigma_z^2 I)$ yields the KL term as
\begin{multline}\label{eq:KL_gaussian}
D_{\mathrm{KL}}=\log \frac{|\sigma_z^2 I|}{|\Sigma_\phi(x)|}+ \operatorname{tr}\left((\sigma_z^2 I)^{-1} \Sigma_\phi(x)\right)\\
+ \mu_\phi(x)^\dagger (\sigma_z^2 I)^{-1}\mu_\phi(x)- d.
\end{multline}

{\color{black}
	For a minibatch $\mathcal{B}$ of size $B$, the empirical
	variational upper bound on the representation rate is estimated as
	\begin{align}
		\widehat C_{\mathrm{rep}}
		=
		\frac{1}{B\ln 2}
		\sum_{i\in\mathcal{B}}
		D_{\mathrm{KL}}
		\left(
		p_{\phi}(z|x_i)\Vert p(z)
		\right),
		\label{eq:empirical_rep_rate}
	\end{align}
where the factor $1/\ln 2$ converts nats into bits per complex input symbol. Then, the dual variable can be updated by
	\begin{align}
		\beta^{(t+1)}
		=
		\left[
		\beta^{(t)}
		+
		\eta_{\beta}
		\left(
		\widehat C_{\mathrm{rep}}^{(t)}
		-
		C_{\mathrm{rep}}
		\right)
		\right]^+,
		\label{eq:dual_update}
	\end{align}
	where $\eta_{\beta}>0$ is the dual learning rate and
	$[a]^+=\max(a,0)$. Hence, $\beta$ increases when the measured
	representation rate exceeds the target budget and decreases when
	the constraint is inactive.
}


{\color{black}
	The Gaussian parameterization in \eqref{eq:reparam} provides a
	learning-based stochastic encoder that is compatible with the
	Gaussian test channel benchmark. However, a general neural encoder
	may learn an input mean and covariance and therefore need
	not exactly reduce to the additive model $Z=X+W_z$. The latter is
	recovered only under the restricted parameterization
	$\mu_\phi(x)=x$ and $\Sigma_\phi(x)=N_z I$.
	
	{\color{black}
		The projected dual update is designed to drive the empirical
		variational rate toward the prescribed budget when the constraint is active. Since the neural optimization problem is nonconvex, the
		effectiveness of the procedure is evaluated empirically in
		Section~\ref{sec:num}.
	}
	For the
	Gaussian test channel benchmark, the corresponding
	representation noise level is
	\begin{align}
		N_z^\star
		=
		\frac{P}{2^{C_{\mathrm{rep}}}-1}.
	\end{align}
	Hence, the Gaussian test channel serves as the
	information-theoretic benchmark, whereas a finite neural encoder may
	deviate from this ideal structure because of its input-dependent
	statistics, finite model capacity, and training dynamics. The ability of the learned encoder to satisfy prescribed
	representation budgets is evaluated numerically in
	Section~\ref{sec:num}.
}


\subsection{Training Algorithm}
Algorithm~\ref{alg:VIB} summarizes the learning process. The neural network encoder receives input symbols $x$ and produces the latent variable $z$. The decoder $p_\psi(m|z)$ performs message reconstruction during training, ensuring that $Z$ retains the information needed for communication decoding. Because $Z$ is fed into the physical ISAC model, the learned latent representations are directly compatible with the representation-constrained FAS-aided ISAC system. The term $\text{CE}(m,\psi(z))$ denotes the standard cross-entropy (CE) loss between the true message label $m$ and the decoder output distribution $\psi(z)$, namely $\text{CE}(m,\psi(z)) = -\log p_\psi(m|z)$.

\begin{algorithm}[t]
\caption{VIB Training for the Representation-Constrained FAS-aided ISAC  Encoder}\label{alg:VIB}
{\color{black}\begin{algorithmic}[1]
	\STATE Initialize encoder parameters $\phi$, decoder parameters
	$\psi$, dual variable $\beta\geq0$, and target budget
	$C_{\mathrm{rep}}$.
	\FOR{each training minibatch}
	\STATE Sample symbols $x$ and corresponding messages $m$.
	\STATE Compute encoder statistics
	$\mu_{\phi}(x)$ and $\Sigma_{\phi}(x)$.
	\STATE Sample $z$ using the reparameterization in
	\eqref{eq:reparam}.
	\STATE Evaluate the task loss
	$\mathcal{L}_{\mathrm{task}}$.
	\STATE Compute the minibatch empirical variational rate estimate 
	$\widehat C_{\mathrm{rep}}$ using \eqref{eq:empirical_rep_rate}.
	\STATE Form
	$\mathcal{L}_{\mathrm{VIB}}
	=
	\mathcal{L}_{\mathrm{task}}
	+
	\beta(\widehat C_{\mathrm{rep}}-C_{\mathrm{rep}})$.
	\STATE Update $(\phi,\psi)$ by gradient descent.
	\STATE Update $\beta$ using \eqref{eq:dual_update}.
	\ENDFOR
\end{algorithmic}}
\end{algorithm}

{\color{black}
	During training, the encoder parameters are optimized jointly
	with the decoder parameters, while the dual update adjusts
	$\beta$ to enforce the prescribed representation budget.
	The empirical variational upper bound
	$\widehat C_{\mathrm{rep}}$ is evaluated from the trained encoder
	and subsequently used to assess budget calibration relative to the
	Gaussian test channel benchmark.
}

\subsection{Practical Implementation}

The VIB encoder is introduced to provide an approximate realization
of the information bottleneck model considered in the theoretical
analysis. Specifically, the encoder and decoder networks are trained
jointly using stochastic gradient optimization to minimize the
VIB objective. {\color{black}
	The achieved empirical variational rate depends on the encoder architecture,	latent dimension, prior distribution, optimization procedure, and the dual variable $\beta$. Consequently, there is no architecture independent closed-form mapping from $\beta$ to
	$C_{\mathrm{rep}}$. 
	The achieved empirical variational rate must instead be measured
	after training using \eqref{eq:empirical_rep_rate}, and the encoder is regarded as satisfying the prescribed budget when
	$|\widehat{C}_{\mathrm{rep}}-C_{\mathrm{rep}}|$
	remains within the adopted numerical tolerance.} 
	 Moreover, while the Gaussian prior assumption aligns with the
theoretical bottleneck model, practical neural implementations may
exhibit some mismatch due to finite model capacity and training
dynamics. Nevertheless, the Gaussian assumption provides a tractable
approximation that captures the principal effect of the representation
constraint. The VIB implementation is therefore intended as an
illustrative bridge between the theoretical bottleneck model and a
learning-based realization, rather than as a fully optimized
end-to-end ISAC architecture.

\subsection{Practical Implications}
Our theoretical results provide useful insights for the design and evaluation of practical ISAC systems. In particular, the characterized capacity-distortion trade-off establishes fundamental performance limits under representation information constraints, which can serve as benchmarks for assessing the efficiency of real implementations. From a system design perspective, the analysis reveals the critical role of the physical size of the FAS in enhancing spatial diversity. 
Specifically, increasing the antenna region enables the receiver to
exploit more favorable post-encoder channel realizations, thereby
reducing the physical channel penalty that remains after
representation-constrained encoding. This spatial gain does not
increase $I(X;Z)$ or recover information discarded by the encoder;
rather, it allows the received performance to approach the
representation-limited bounds more closely.  Furthermore, the results provide guidance on the selection of key system parameters. For example, the trade-off analysis can inform the choice of the number of antenna ports, the physical size of the FAS, and the required representation budget to achieve target communication and sensing performance levels. In this sense, the framework enables a principled co-design of communication, sensing, and representation resources. 

The proposed model also offers insights into the design of learning-based encoders. The representation budget can be interpreted as a design requirement on the latent representation dimension or quantization level in practical neural network architectures. {\color{black}However, the representation budget should not be viewed as a direct measure of neural network size or complexity. Different architectures may realize the same representation budget. Instead, it provides an operational quantity for comparing encoder implementations at a common representation rate and for determining whether a trained encoder meets a required system level information constraint.} Moreover, while the present work focuses on an analytically tractable communication-centric port selection rule, the derived results can serve as a baseline for evaluating more advanced strategies, such as joint communication--sensing selection policies and learning-based control of antenna configurations. Therefore, the proposed framework provides a foundation for bridging information-theoretic analysis and practical system design in reconfigurable antenna-aided ISAC systems.

\section{Numerical Results}\label{sec:num}
Here, we present numerical results to validate the theoretical analysis and illustrate the impact of the  representation bottleneck and FAS on ISAC performance. The simulation parameters are summarized in Table \ref{tab:sim}, including the transmit power, noise variance, representation budget, and FAS configuration. {\color{black} Unless otherwise stated, the numerical results use normalized large-scale gains $\beta_c=\beta_s=1$. Therefore, the reported SNRs are considered as reference average received SNRs, while the FAS lengths $W\in\{0.5\lambda,2\lambda,8\lambda\}$ describe the physical aperture size normalized by the carrier wavelength.} 
{\color{black} For the subsequent sensing and rate-distortion results, we adopt $\rho_{cs}=0.5$ as the main operating case to represent a partial similarity between the communication and bistatic sensing channels, while $\rho_{cs}=1$ is additionally shown as an idealized common spatial fading benchmark.}  Benchmarks based on single-input single-output (SISO) and MIMO systems are included. These benchmarks are not intended to provide a strict performance comparison under equivalent hardware constraints, but rather to illustrate the impact of spatial reconfigurability and scaling behavior of the FAS-enabled ISAC receiver with different physical lengths. In particular, FAS configurations with a large number of ports and extended physical size inherently offer increased spatial DoF. Therefore, the observed performance gains should be viewed as indicative of the potential benefits of spatial reconfigurability, rather than as a direct comparison with conventional MIMO systems of similar complexity.

\begin{table}[t]
\centering
\caption{Simulation parameters.}\label{tab:sim}
\begin{tabular}{lcc}
\toprule
Quantity & Symbol & Value \\
\midrule
Average  SNR & $\rho=P/N_0$ & $10$ dB\\
Normalized noise variance & $N_0$ & $0.1$ \\
Sensing parameter variance & $\sigma_\theta^2$ & $1$\\
Normalized Wavelength & $\lambda$ & $1$ \\
Number of ports & $L$ & $256$ \\
MIMO & $(N_t,N_r)$ & $(2,2)$  \\
Length of fluid antenna & $W$ & $\{0.5\lambda, 2\lambda, 8\lambda\}$ \\
Target representation budgets & $C_{\mathrm{rep}}$ & $\{2,4,6,\infty\}$  \\
\color{black}Equivalent representation noise ratios & \color{black}$N_z^\star/P$ & \color{black}$\left\{\frac{1}{3},\frac{1}{15},\frac{1}{63},0\right\}$\\
\bottomrule
\end{tabular}
\end{table}

\begin{table}[t]
	\centering
	\caption{{\color{black}Calibration of the learned VIB encoder under different prescribed representation budgets. 
	}}\label{tab:vib}
	\footnotesize
	\setlength{\tabcolsep}{4pt}
	{\color{black}\begin{tabular}{c c c c}
			\hline
			Target $C_{\mathrm{rep}}$
			&
			Empirical 
			rate $\widehat C_{\mathrm{rep}}$
			&
			Test accuracy
			&
			Steady state $\bar{\beta}$
			\\
			\hline
			$2$ &
			$1.997 \pm 0.020$ &
			$5.18 \pm 0.09\%$ &
			$0.678 \pm 0.001$
			\\
			$4$ &
			$4.007 \pm 0.055$ &
			$15.71 \pm 0.35\%$ &
			$0.655 \pm 0.003$
			\\
			$6$ &
			$6.023 \pm 0.027$ &
			$42.59 \pm 0.83\%$ &
			$0.580 \pm 0.003$
			\\
			\hline
	\end{tabular}}
\end{table}

{\color{black}\subsection{Calibration of the VIB Encoder}
	\label{subsec:vib_calibration}}
{\color{black}
	We first examine whether the constrained VIB training procedure can
	enforce prescribed representation budgets. We consider
	$256$ equiprobable messages associated with a fixed $256$ point complex codebook whose entries are drawn once from $\mathcal{CN}(0,1)$ and subsequently normalized to unit average
	power. The encoder and decoder each contain two fully connected hidden layers with $128$ neurons and ReLU activations. The encoder produces the mean and variance of a $1$D proper complex Gaussian latent variable,
	while the decoder performs $256$ class message classification.
}

{\color{black}
	For each target
	$C_{\mathrm{rep}}\in\{2,4,6\}$ bits per complex input symbol, a
	separate encoder-decoder pair is trained using Algorithm~\ref{alg:VIB}. Training
	uses minibatches of size $512$ and $50$ minibatches per epoch. The
	encoder and decoder are optimized using Adam with learning rates of
	$10^{-3}$ for $250$ epochs. The dual variable is initialized as
	$\beta_0=0.1$ and updated using the dual learning rate
	$\eta_\beta=5\times10^{-4}$, followed by projection onto $[0,100]$.
	The latent prior is the  $1$-D proper complex Gaussian
	distribution $\mathcal{CN}(0,1)$.
	Each experiment is repeated over five independent random
	initializations. The empirical variational upper bound on the
	representation rate is evaluated using
	\eqref{eq:empirical_rep_rate}, while the steady state dual variable
	$\bar{\beta}$ is obtained by averaging $\beta$ over the final $20$
	training epochs. Test classification accuracy is evaluated using
	$50$ independent latent samples per message.}
{\color{black}
Moreover, the test classification accuracy increases from
$5.18\%$ to $42.59\%$ as the representation budget increases, while
the corresponding CE loss decreases from $4.1753$ to
$1.4994$. This behavior confirms that the encoder preserves
progressively more task relevant information under less restrictive
representation budgets. The reported
$\widehat C_{\mathrm{rep}}$ values are empirical variational upper
bounds on $I(X;Z)$ and should not be interpreted as exact
MI estimates.
}

{\color{black}
As reported in Table~\ref{tab:vib} and demonstrated in Fig.~\ref{fig:vib_calibration}, the achieved variational rates closely track their prescribed targets. Specifically, the trained encoder achieves $1.997\pm0.020$, $4.007\pm0.055$, and $6.023\pm0.027$ bits per complex input symbol for target budgets of $2$, $4$, and $6$ bits, respectively. In every case, the absolute difference between the mean achieved rate and its prescribed value remains below the adopted tolerance of $0.05$ bit. These results demonstrate that the prescribed representation budget can be closely enforced for the considered VIB architecture through the projected dual update, rather than being used solely as an unvalidated numerical parameter.
}

\begin{figure}[t]
\centering
\includegraphics[width=0.9\columnwidth]{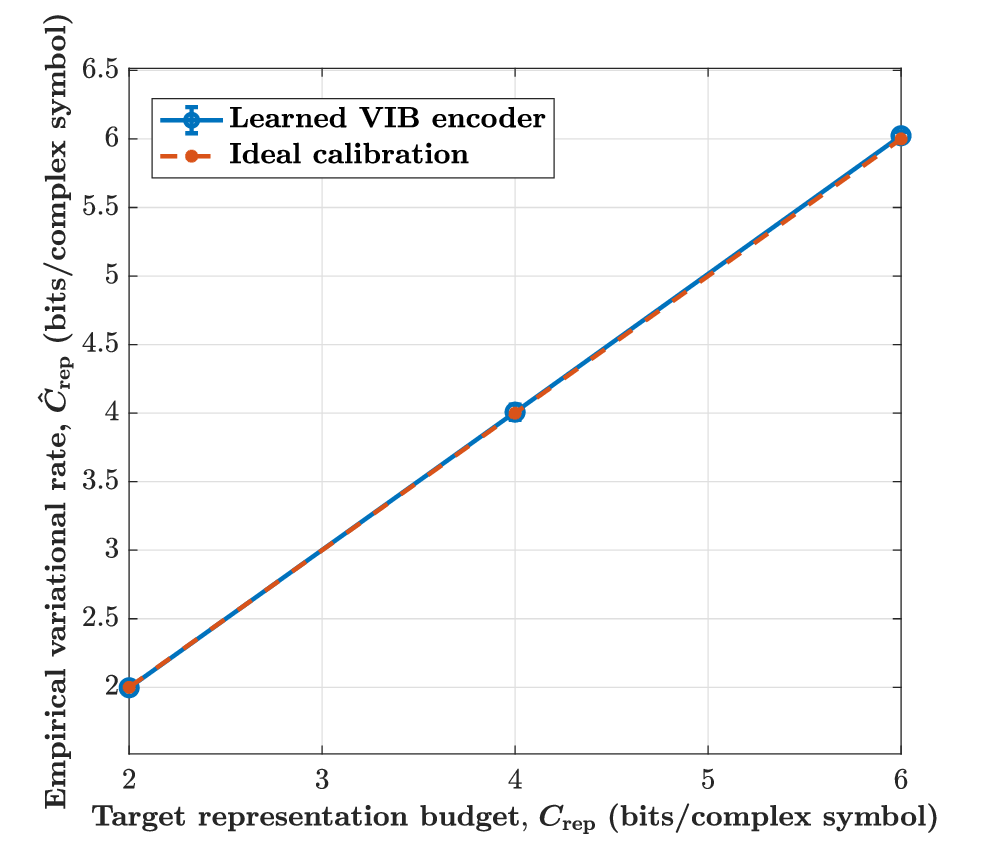}
\caption{{\color{black}Prescribed representation budgets and the corresponding empirical variational rates of the learned VIB encoder.}}\label{fig:vib_calibration}
\end{figure}

\begin{figure*}[t]
	\centering
	\subfloat[$C_{\mathrm{rep}}=2$ bits/complex symbol.]
	{
		\includegraphics[width=0.31\textwidth]
		{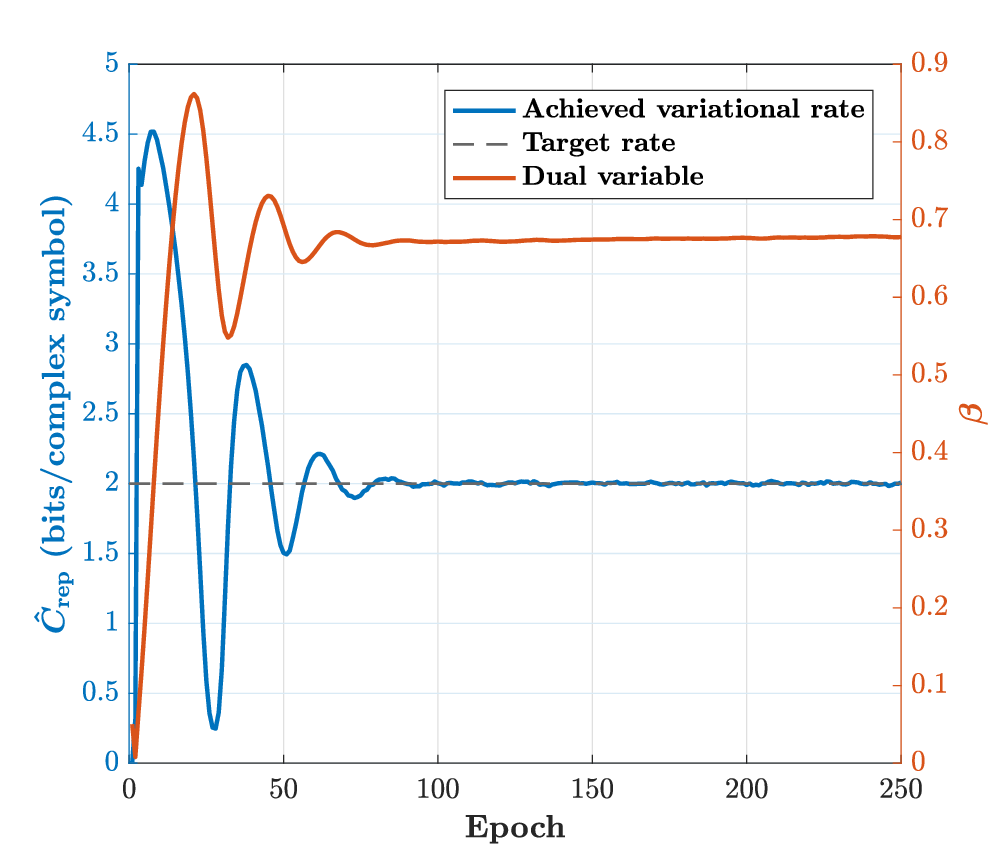}
		\label{fig:vib_trajectory_2}
	}
	\hfill
	\subfloat[$C_{\mathrm{rep}}=4$ bits/complex symbol.]
	{
		\includegraphics[width=0.31\textwidth]
		{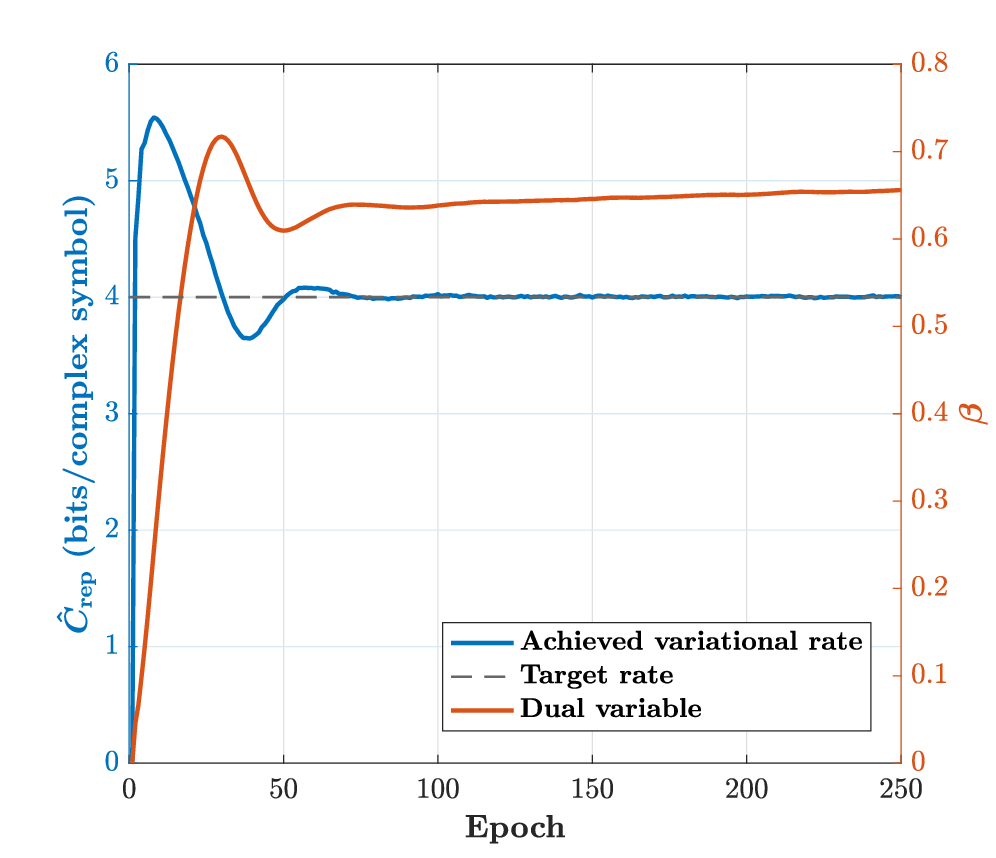}
		\label{fig:vib_trajectory_4}
	}
	\hfill
	\subfloat[$C_{\mathrm{rep}}=6$ bits/complex symbol.]
	{
		\includegraphics[width=0.31\textwidth]
		{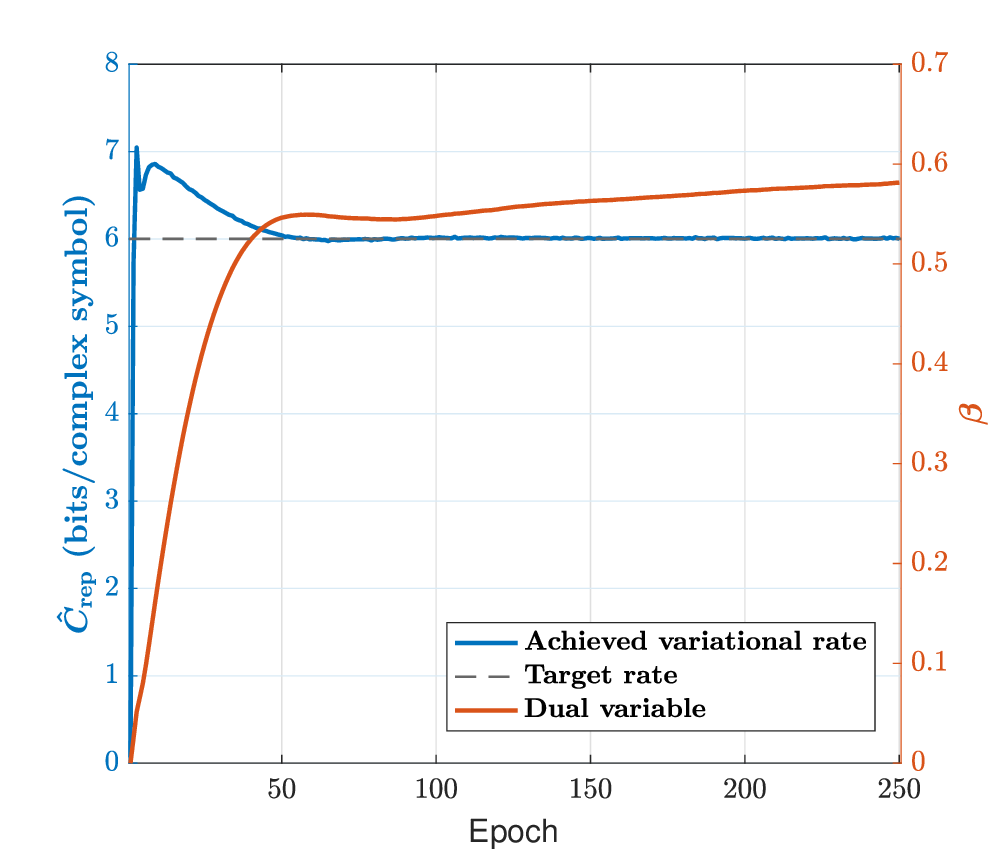}
		\label{fig:vib_trajectory_6}
	}
\caption{{\color{black}Convergence behavior of the constrained VIB training procedure for different prescribed representation budgets. Each panel shows the evolution of the empirical variational rate (left y-axis) and the corresponding dual variable (right y-axis), averaged over five independent training runs.}}\label{fig:vib_trajectories}
\end{figure*}

{\color{black}
The training dynamics are shown in Fig.~\ref{fig:vib_trajectories}. For all three prescribed budgets, the empirical variational rate initially undergoes a transient adjustment and then stabilizes in the vicinity of the target value. The dual variable evolves simultaneously to penalize rates above the prescribed budget and relax the constraint when the measured rate is below it. The similar convergence behavior observed for $C_{\mathrm{rep}}=2$, $4$, and $6$ confirms that the calibration is not specific to a single operating point and that the constrained training procedure remains stable across the considered representation regimes.
}

{\color{black}
	Fig.~\ref{fig:vib_accuracy} shows that the test classification
	accuracy increases from $5.18\%$ at
	$C_{\mathrm{rep}}=2$ bits to $15.71\%$ and $42.59\%$ at budgets of
	$4$ and $6$ bits, respectively. This monotonic improvement demonstrates
	that relaxing the representation information constraint allows the encoder to
	retain progressively more task relevant information. It also confirms
	that satisfaction of the representation budget is not achieved through
	a trivial collapse of the latent representation. The classification
	accuracy is included as a task-fidelity diagnostic for the VIB
	implementation and is not intended to replace the communication rate
	and sensing-distortion metrics used in the subsequent 
	analysis.
}
\begin{figure}[t]
	\centering
	\includegraphics[width=0.9\columnwidth]{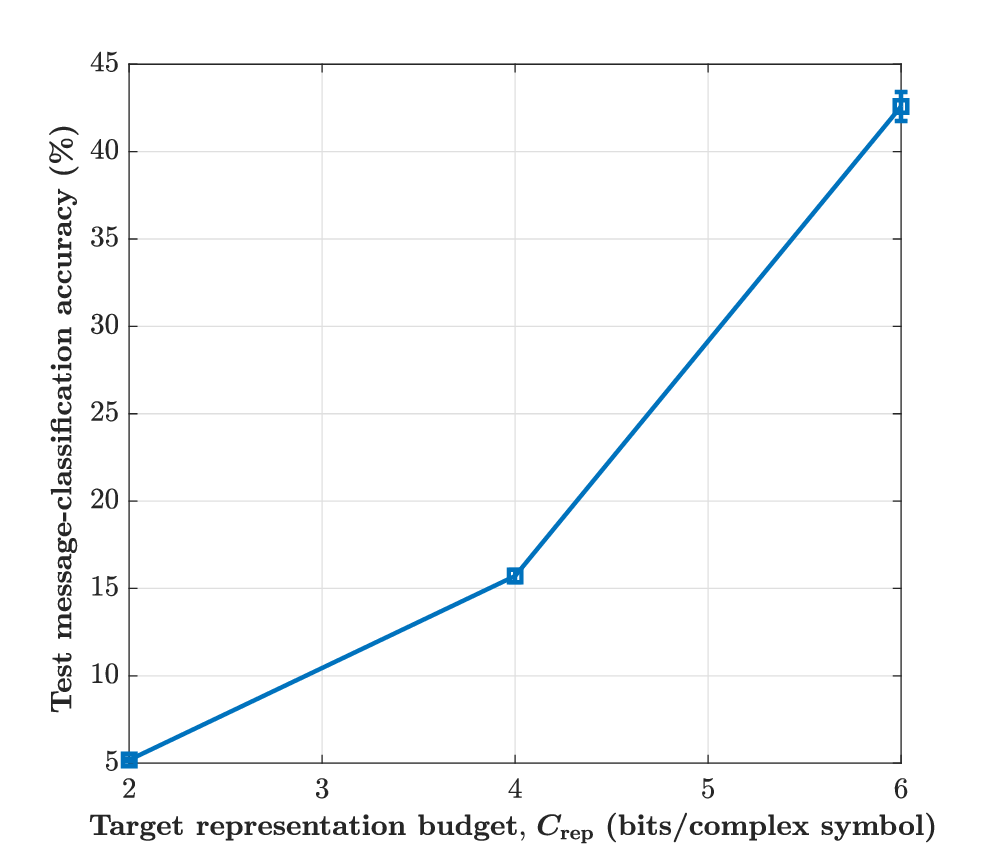}
	\caption{\textcolor{black}{Test message classification accuracy of the learned VIB encoder versus the prescribed representation budget. 
	}}
	\label{fig:vib_accuracy}
\end{figure}

{\color{black}\subsection{Impact of the Representation Budget and FAS Length}
	\label{subsec:budget_performance}
	The values $C_{\mathrm{rep}}\in\{2,4,6\}$ bits per complex input
	symbol are used to represent low-, moderate-, and high-quality
	representation regimes, respectively. Under the Gaussian test channel
	benchmark, these budgets correspond to normalized representation noise
	levels
	$
	\frac{N_z^\star}{P}
	\in
	\left\{
	\frac{1}{3},
	\frac{1}{15},
	\frac{1}{63}
	\right\}
	$. 
	Furthermore, the calibration results in
	Table~\ref{tab:vib} demonstrate that these prescribed
	budgets can be closely realized by the trained VIB encoder. Hence, the
	selected values define quantitatively specified and practically
	enforceable representation operating points.} 
Fig.~\ref{fig:r_ai} illustrates the achievable communication rate
$R$ as a function of the representation budget $C_{\mathrm{rep}}$
for several representative architectures, including conventional
SISO and MIMO links as well as the proposed FAS-based ISAC receiver with different physical aperture lengths. Since the adopted port selection policy and the communication rate depend only on the communication channel, the curves in Fig.~\ref{fig:r_ai} are independent of the communication--sensing correlation coefficient
$\rho_{cs}$. Across all configurations, the achievable rate increases monotonically with $C_{\mathrm{rep}}$, because a larger representation budget reduces the equivalent representation noise level and allows the physical channel to support a larger communication rate. The SISO architecture exhibits the lowest rate over the considered range, whereas the $2\times2$ MIMO benchmark provides a substantial improvement and approaches the representation-limited upper bound $R=C_{\mathrm{rep}}$ more closely.
The FAS curves demonstrate the effect of increasing the physical
aperture length $W$ under the same representation constraint. For the small aperture $W=0.5\lambda$, the achievable rate is comparable to, and over part of the considered range lower than, that of the $2\times2$ MIMO benchmark. As the aperture increases to $W=2\lambda$ and $W=8\lambda$, the larger number of effective spatial DoF improves the upper tail of the selected communication channel gain, resulting in progressively higher achievable rates. In particular, the $W=8\lambda$ configuration provides the largest rate among the considered architectures and approaches the representation-limited ceiling most closely. The advantage of a larger FAS aperture becomes more visible as $C_{\mathrm{rep}}$ increases, since the representation bottleneck becomes less dominant and the performance is increasingly determined by the selected post-encoder physical channel gain.
\begin{figure}[t]
\centering
\includegraphics[width=.9\columnwidth]{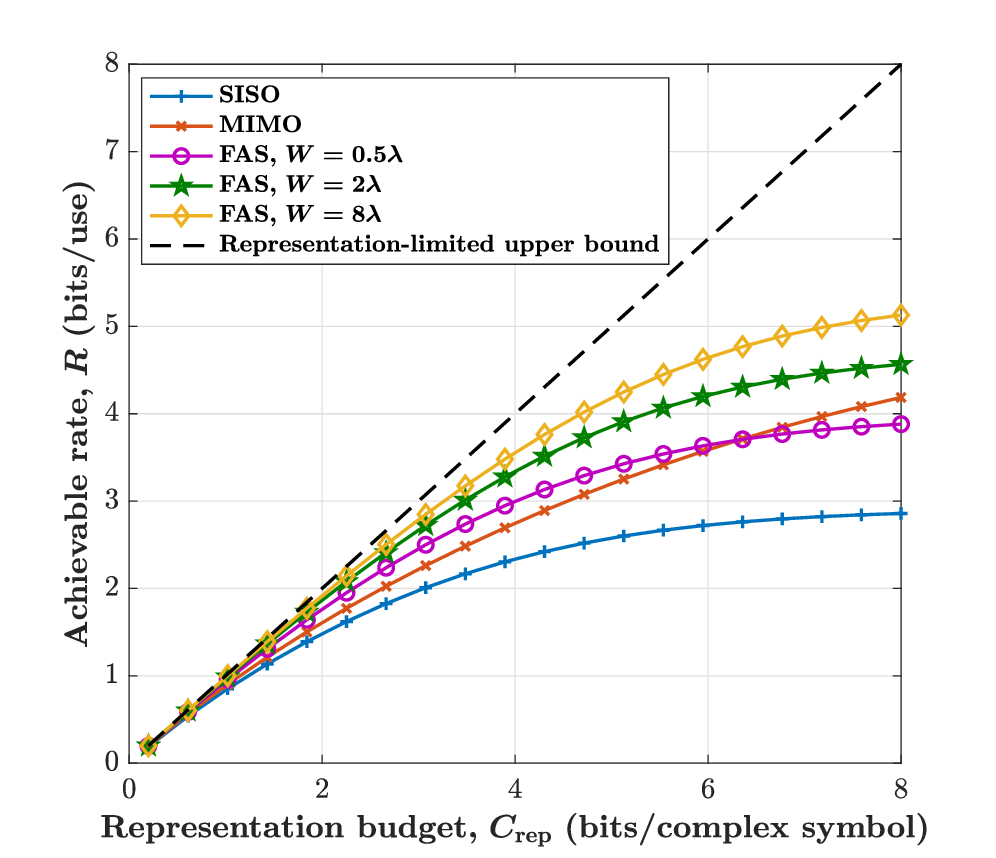}
\caption{Achievable communication rate $R$ versus the prescribed
	representation budget $C_{\mathrm{rep}}$.}\label{fig:r_ai}
\end{figure}

\begin{figure}[t]
\centering
\includegraphics[width=.9\columnwidth]{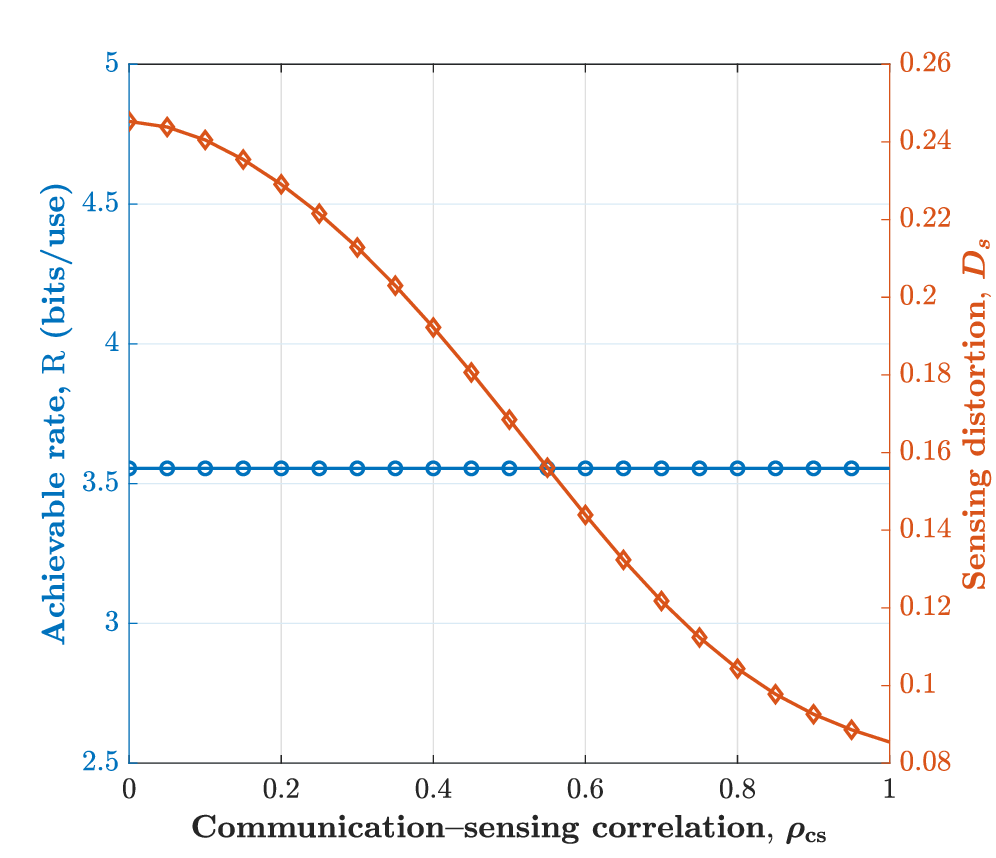}
\caption{\textcolor{black}{Impact of the communication-sensing spatial correlation coefficient $\rho_{cs}$ on the achievable communication rate and
sensing distortion for $C_{\mathrm{rep}}=4$ bits/complex symbol, $W=8\lambda$, and communication-centric FAS port selection.}}\label{fig:rho_cs}
\end{figure}

{\color{black}Before examining the sensing performance under a fixed communication-sensing similarity level, Fig.~\ref{fig:rho_cs} illustrates the effect of $\rho_{cs}$ on the achievable communication rate and sensing distortion. The results are obtained for $C_{\mathrm{rep}}=4$ bits/complex symbol and $W=8\lambda$, while the FAS port is selected solely according to the communication channel gain.  The achievable communication rate remains constant as $\rho_{cs}$ varies. This behavior is expected because both the selected port index and the communication rate depend only on $\mathbf h_c$ and are therefore independent of the sensing channel. In contrast, the sensing distortion decreases monotonically as $\rho_{cs}$ increases. When $\rho_{cs}=0$, the communication and sensing spatial channels are independent, so communication-centric port selection provides no sensing selection gain. As the correlation increases, the port selected from the communication channel becomes progressively more favorable for the sensing channel, thereby increasing the selected sensing gain and reducing the distortion. The case $\rho_{cs}=1$ represents the limit case for which the communication selection gain is fully transferred to sensing. These results explicitly demonstrate that the sensing benefit of the adopted communication-centric port selection policy depends on the communication-sensing spatial correlation. 
}

{\color{black}
	Fig.~\ref{fig:d_ai} presents the sensing distortion $D_s$ as a function of the representation budget $C_{\mathrm{rep}}$ for SISO, MIMO, and FAS receivers with different physical aperture lengths. To illustrate the impact of the communication-sensing spatial correlation under the adopted communication-centric FAS port selection policy, Fig.~\ref{fig:d_ai}(a) considers the moderate correlation case $\rho_{cs}=0.5$, which is used as the main sensing benchmark in this work, whereas Fig.~\ref{fig:d_ai}(b) shows the idealized common spatial fading case $\rho_{cs}=1$. Across all architectures and both correlation settings, the sensing distortion decreases monotonically with $C_{\mathrm{rep}}$. This behavior follows from the reduction in the equivalent representation noise level as the prescribed budget increases, which improves the quality of the post-encoder sensing observation. For small representation budgets, particularly below approximately $2$ bits/complex symbol, the representation bottleneck dominates the system performance and the distortion levels of the different architectures remain relatively close. As $C_{\mathrm{rep}}$ increases, the physical sensing channel becomes increasingly influential and the differences among the receiver architectures become more pronounced.
	
	For the moderate correlation case $\rho_{cs}=0.5$ in Fig.~\ref{fig:d_ai}(a), the FAS configurations provide a clear improvement over SISO, and their sensing performance improves progressively as the physical aperture increases from $W=0.5\lambda$ to $W=8\lambda$. This behavior is a direct consequence of the communication-centric port selection rule: the selected port maximizes the communication channel gain and, for $\rho_{cs}=0.5$, is only partially favorable for the sensing channel. In contrast, Fig.~\ref{fig:d_ai}(b) considers the idealized case $\rho_{cs}=1$, for which the communication and sensing spatial channels are fully aligned and the communication selection gain is fully transferred to sensing.  In this case, the sensing advantage of increasing the FAS aperture becomes substantially stronger, and the large aperture FAS outperforms the MIMO benchmark over a significant portion of the considered representation-budget range. Therefore, the two correlation cases illustrate how strongly the sensing performance of communication-centric FAS port selection depends on the degree of spatial alignment between the communication and sensing channels. The dotted curve in both panels represents the representation-limited lower bound $\sigma_\theta^2/2^{C_{\mathrm{rep}}}$, corresponding to an arbitrarily large physical sensing gain while the prescribed representation budget remains fixed. In the moderate-correlation case, the gap to this bound is jointly affected by the finite FAS aperture and the incomplete transfer of the communication-selection gain to sensing, whereas for $\rho_{cs}=1$ the remaining gap is primarily due to the finite physical aperture.
}
\begin{figure}[t]
	\centering
	\subfloat[$\rho_{cs}=0.5$]{
		\includegraphics[width=0.9\columnwidth]{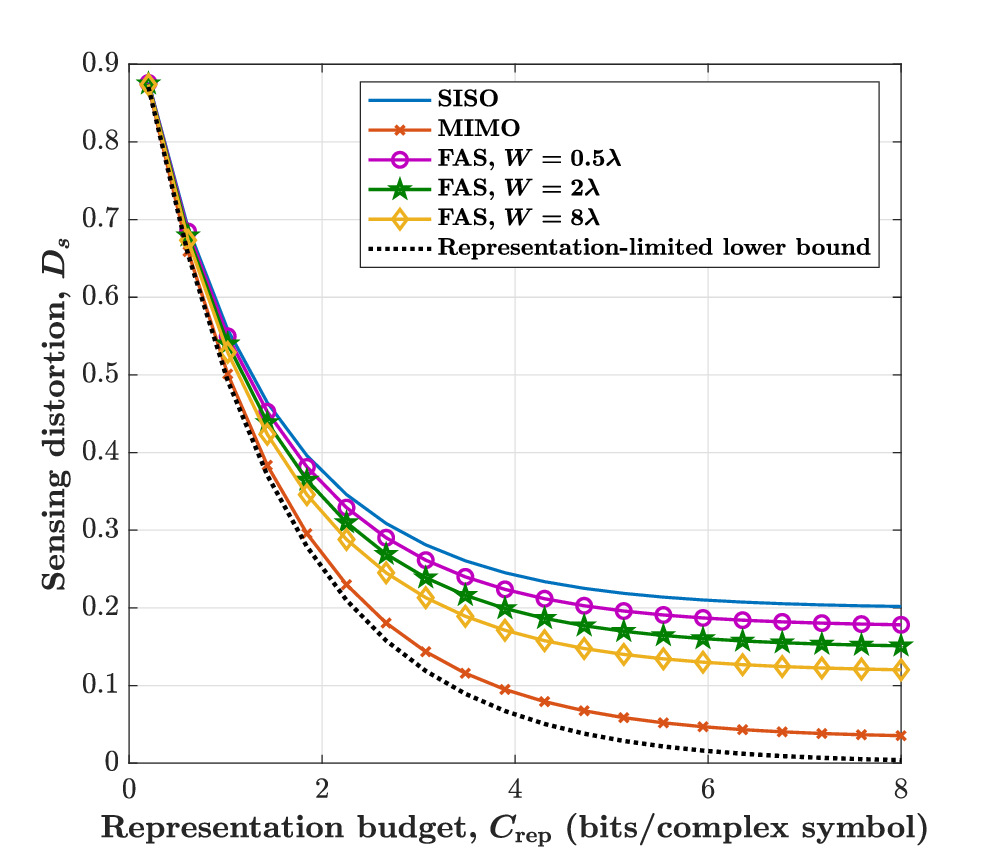}
		\label{fig:d_ai_a}}
	\hfill
	\subfloat[$\rho_{cs}=1$]{
		\includegraphics[width=.9\columnwidth]{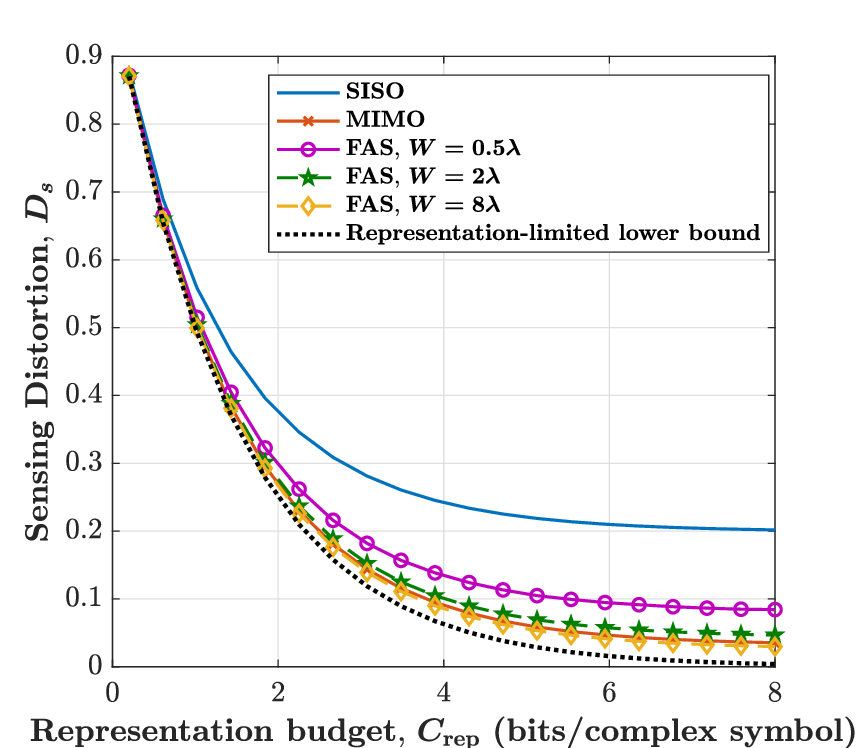}
		\label{fig:d_ai_b}}
	\caption{\textcolor{black}{Sensing distortion $D_s$ versus the prescribed representation budget $C_{\mathrm{rep}}$: (a) moderate communication-sensing spatial correlation with $\rho_{cs}=0.5$; (b) idealized common spatial fading with $\rho_{cs}=1$.}}
	\label{fig:d_ai}
\end{figure}

{\color{black}
	Fig.~\ref{fig:d_r} presents representative numerical operating curves between the achievable communication rate $R$ and the sensing distortion $D_s$ for the considered representation-constrained ISAC architecture with $W=8\lambda$. Fig.~\ref{fig:d_r}(a) considers the moderate communication-sensing spatial correlation $\rho_{cs}=0.5$, which is adopted as the main operating case, whereas Fig.~\ref{fig:d_r}(b) shows the idealized common spatial fading case $\rho_{cs}=1$. The operating points are generated by sweeping a normalized power allocation factor $\alpha\in[0.05,0.95]$, where the useful communication and sensing power terms are taken as $\alpha P$ and $(1-\alpha)P$, respectively, in \eqref{eq-snrc} and \eqref{eq-snrs}. Increasing $\alpha$ allocates more useful power to communication, thereby increasing the achievable rate while also increasing the sensing distortion.
	
	For the smallest representation budget, $C_{\mathrm{rep}}=2$ bits/complex symbol, the operating curves remain strongly constrained by the representation bottleneck. As the representation budget increases to $C_{\mathrm{rep}}=4$ and $6$ bits/complex symbol, the curves shift toward higher achievable rates and lower sensing distortions, reflecting the reduction in the equivalent representation noise. In the moderate correlation case of Fig.~\ref{fig:d_r}(a), the large aperture FAS benefits from its favorable communication-selected channel gain and extends to a larger maximum communication rate. However, because the selected port is determined solely from the communication channel, for $\rho_{cs}=0.5$ the communication selection gain is only partially transferred to the bistatic sensing channel. In contrast, for the idealized case $\rho_{cs}=1$ in Fig.~\ref{fig:d_r}(b), the communication and sensing spatial channels are fully aligned, so the communication selection gain is fully transferred to sensing. Consequently, the FAS rate-distortion operating curves become substantially more favorable, particularly for the larger representation budgets, and the FAS  outperforms the MIMO benchmark over the corresponding operating region. The two panels therefore illustrate the strong dependence of the achievable communication-sensing operating region on the communication-sensing spatial correlation under communication-centric FAS port selection.
	
	The curves labeled $C_{\mathrm{rep}}\to\infty$ correspond to the idealized case in which the representation encoder introduces no additional representation noise. Comparing these curves with their finite-$C_{\mathrm{rep}}$ counterparts also reveals a saturation effect: beyond approximately $C_{\mathrm{rep}}=6$ bits/complex symbol, further increases in the representation budget yield relatively modest improvements as the system performance becomes increasingly governed by the physical channel and thermal noise.
}


\begin{figure}[t]
	\centering
	\subfloat[$\rho_{cs}=0.5$]{
		\includegraphics[width=.9\columnwidth]{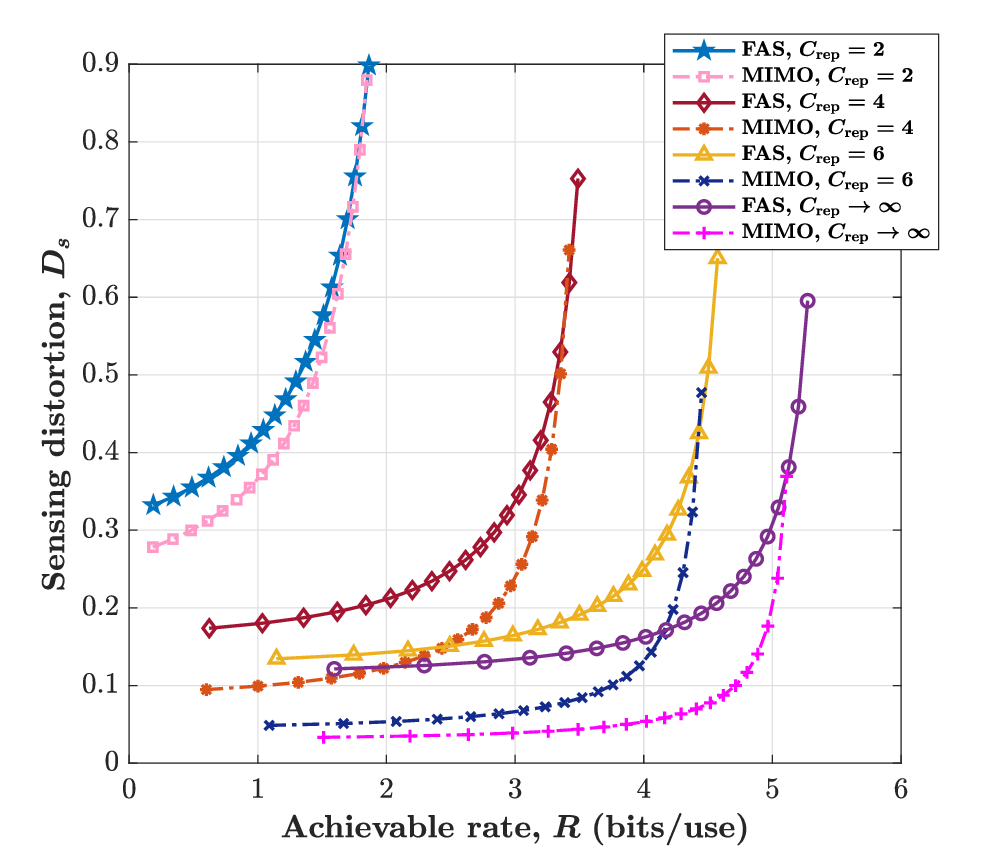}
		\label{fig:d_r_a}}
	\hfill
	\subfloat[$\rho_{cs}=1$]{
		\includegraphics[width=.9\columnwidth]{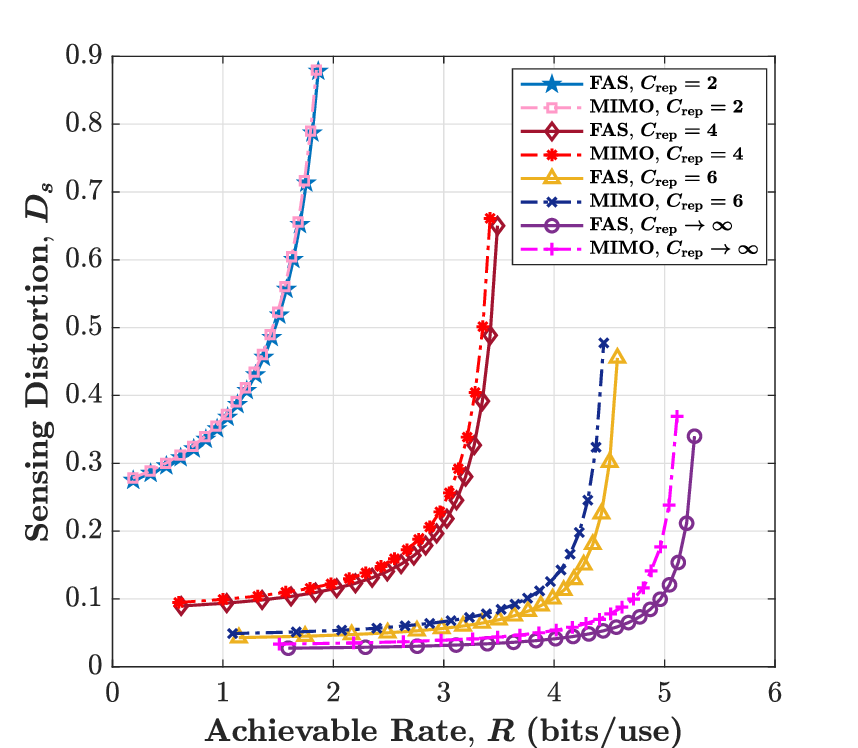}
		\label{fig:d_r_b}}
	\caption{\textcolor{black}{Representative communication-sensing rate-distortion operating curves for $W=8\lambda$: (a) moderate communication-sensing spatial correlation $\rho_{cs}=0.5$; and (b) idealized common spatial fading $\rho_{cs}=1$.}}
	\label{fig:d_r}
\end{figure}

{\color{black}\subsection{Validation of the Finite-$L$ Analysis}}
\label{subsec:finiteL_validation}
Finally, Fig.~\ref{fig:te_sim} shows the validation of the proposed finite-$L$ FAS performance characterization under the Jakes' spatial correlation model. This figure compares the communication rate and sensing distortion predicted by the finite-$L$ theory, obtained by evaluating \eqref{eq:R_region}--\eqref{eq:D_region} using the exact distribution of the selected port gain $\gamma_c^\star$, with Monte Carlo simulations of a FAS with $L = 256$ ports. The theoretical and simulated curves coincide across the entire range of $C_{\mathrm{rep}}$, confirming that the analytical model accurately captures the behavior of spatially correlated FAS arrays. This close agreement verifies that the representation noise interpretation of the representation bottleneck remains valid under the Jakes correlated finite aperture model and that the proposed finite-$L$ analysis provides an accurate prediction of the ISAC performance.

\begin{figure}[t]
\centering
\includegraphics[width=.9\columnwidth]{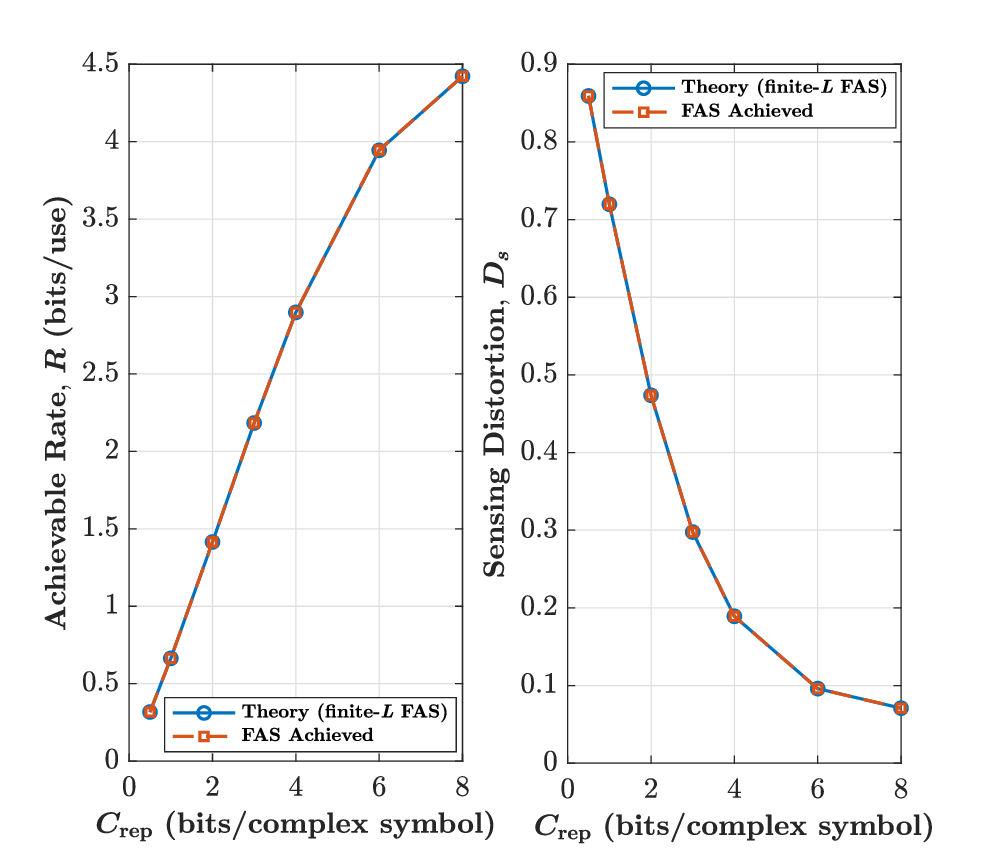}
\caption{Validation of the finite-$L$ theoretical expressions against the achieved communication rate and sensing distortion under the common spatial fading benchmark $\rho_{cs}=1$.}\label{fig:te_sim}
\end{figure}

\section{Conclusion}\label{sec:con}
{\color{black}This work characterized the information-theoretic performance of FAS-assisted ISAC under a pre-channel representation constraint. Under
the Gaussian test channel model, the representation constraint appears as an equivalent post-encoder impairment that imposes finite
communication and sensing performance ceilings. We derived the corresponding capacity-distortion bounds and characterized the effect of finite-aperture Jakes correlated FAS port selection. The FAS receiver does not alter $I(X;Z)$ or recover information discarded by the representation encoder. Instead, increasing the physical FAS aperture improves the selected post-encoder physical channel gain and reduces the remaining thermal noise penalty, allowing communication performance to approach the representation-limited rate ceiling. {\color{black}The same conclusion applies to sensing when the selected sensing gain also increases with the FAS aperture, under the idealized common spatial fading benchmark.}  Finally, the constrained VIB experiment revealed that prescribed representation budgets can be closely enforced through a projected dual update for the considered encoder architecture. The resulting framework provides an information-theoretic benchmark for evaluating FAS-assisted ISAC systems operating with finite representation resources.}

\bibliographystyle{IEEEtran}

\end{document}